\documentclass{revtex-stylesheet}

\begin{document}
\title{An efficient Pauli decomposition algorithm for structured matrices}

\author{Daniel J. Spencer}
\email{djspence@umd.edu}
\affiliation{Joint Center for Quantum Information and Computer Science, NIST/University of Maryland, College Park, MD 20742, USA}
\affiliation{Joint Quantum Institute, NIST/University of Maryland, College Park, MD 20742, USA}
\affiliation{Department of Physics, University of Maryland, College Park, MD 20742, USA}

\author{Kishor Bharti}
\affiliation{Joint Center for Quantum Information and Computer Science, NIST/University of Maryland, College Park, MD 20742, USA}
\affiliation{Department of Computer Science and Institute for Advanced Computer Studies, University of Maryland, College Park, Maryland 20742, USA}

\author{Alexey V. Gorshkov}
\affiliation{Joint Center for Quantum Information and Computer Science, NIST/University of Maryland, College Park, MD 20742, USA}
\affiliation{Joint Quantum Institute, NIST/University of Maryland, College Park, MD 20742, USA}

\date{\today}

\begin{abstract}
    Decomposing classical matrices into linear combinations of Pauli strings is a major bottleneck for end-to-end implementations of near-term quantum algorithms. In this work, we consider a promise version of this Pauli decomposition problem in which the matrix is guaranteed to have support on only $k = \poly(n)$ Pauli strings and is given through classical sparse query access. Existing Pauli decomposition algorithms are designed for the generic, dense problem and do not inherently take advantage of this promised sparsity, so these approaches take time that is exponential in $n$. We present a randomized classical algorithm that does take advantage of this sparsity and recovers the exact Pauli decomposition with success probability at least $1 - \delta$, for any $\delta$. Under the stated access model, the algorithm executes with query and runtime complexity that is polynomial in $n$, $k$, and $\log(1/\delta)$. These results show that, even though finding the Pauli decomposition is exponentially hard for general matrices, it becomes efficiently solvable for matrices that are known to be sparse in the Pauli basis, a regime that is relevant to near-term quantum algorithms operating on structured classical input.
\end{abstract}

\maketitle

\section{Introduction}\label{sec:intro}
Quantum algorithms for applications in chemistry~\cite{lanyon2010towards,cao2019quantum,bauer2020quantum,mcardle2020quantum,motta2021emerging}, machine learning~\cite{lloyd2013quantum,schuld2014introduction,biamonte2017quantum,dunjko2018machine,ramezani2020machine,zaman2023survey,wang2024comprehensive}, and finance~\cite{orus2019quantum,bouland2020prospects,herman2023quantum} often require classically-specified data to be loaded into a representation accessible to a quantum device, leading to the so-called \emph{data-loading problem}~\cite{weigold2021encoding,weigold2021expanding,nakaji2022approximate,marin-sanchez2023quantum}. In many near-term intermediate-scale quantum (NISQ)~\cite{preskill2018quantum,bharti2022noisy} settings, this representation takes the form of a linear combination of unitaries (LCUs), where the unitaries are often taken to be Pauli strings due to their natural compatibility with current quantum hardware. Obtaining this Pauli decomposition can be computationally costly, though this problem is frequently overlooked, with many NISQ algorithms simply \emph{assuming} that the input data are already written in the Pauli basis. When considered in an end-to-end context, finding this decomposition may ruin the supposed speedup before even running the quantum algorithm.

This problem arises in several prominent algorithmic settings. Variational methods such as the variational quantum eigensolver (VQE)~\cite{peruzzo2014variational,kirby2021variational,tilly2022variational} often assume that the input to the problem is a Hamiltonian decomposed into a linear combination of Pauli strings. Other near-term algorithms for optimization~\cite{patel2024variational} and linear algebra~\cite{huang2021near-term,bravo2023variational,pellow2023near} make similar assumptions, where the input matrices (e.g., the matrix $A$ to solve the linear system $Ax = b$) are assumed to be provided as an LCU. For problems in physics and chemistry, such decompositions are readily available from the underlying structure of the problem. However, for more general classically-specified matrices, the Pauli decomposition may not be known in advance and must be explicitly computed. While VQE is seen as one of the leading candidates for a useful near-term algorithm, alternative algorithms~\cite{bharti2021iterative,bharti2022sdp} have been proposed that avoid some of the problems of VQE (e.g., the barren plateau problem~\cite{larocca2025barren}), but these too require the input to be given as a linear combination of Pauli strings. As such, there is a strong necessity for efficient Pauli decomposition algorithms.

All algorithms~\cite{bergholm2018pennylane,gunlycke2020efficient,romero2023paulicomposer,jones2024decomposing,koska2024tree-approach,georges2025pauli} for decomposing a $2^n \times 2^n$ matrix into the Pauli basis scale exponentially in $n$ in the worst case, where the matrix has support on all $4^n$ Pauli basis elements. In this case, there is an immediate worst-case lower bound of $\bigOmega{4^n}$ since all $4^n$ Pauli terms must be recovered. The most straightforward way to obtain the Pauli decomposition of a $2^n \times 2^n$ matrix $M$ is to calculate the overlap of each Pauli string with $M$ to obtain the coefficient. Assuming standard matrix multiplication runtime of $\bigOh{N^3}$ for an $N \times N$ matrix, this requires $4^n$ multiplications of two $2^n \times 2^n$ matrices, for a total time complexity of $\bigOh{32^n}$. Although several existing algorithms~\cite{gunlycke2020efficient,romero2023paulicomposer,jones2024decomposing,koska2024tree-approach,georges2025pauli} make significant improvements over the na\"{i}ve trace-overlap approach, they are designed for the generic dense decomposition problem, where they either compute all $4^n$ Pauli coefficients or require access to the full matrix representation; see~\cref{tab:previous-algorithms}. The best-known algorithm runs in time $\bigOh{n4^n}$ and is implemented in the Python package PennyLane~\cite{bergholm2018pennylane}.

\begin{table}[b]
    \begin{tblr}{colspec={|c|c|c|},hlines,vlines,row{even}={bg=themecolor!50},row{1}={font=\bfseries}}
        Algorithm & \makecell{Runtime complexity}\\
        Direct algorithm & $\bigOh{32^n}$\\
        PennyLane~\cite{bergholm2018pennylane} & $\bigOh{n4^n}$\\
        Gunlycke \etal~\cite{gunlycke2020efficient} & $\bigOh{n4^n}$\\
        PauliComposer~\cite{romero2023paulicomposer} & $\bigOh{8^n}$\\
        Jones~\cite{jones2024decomposing} & $\bigOh{8^n}$\\
        Koska \etal~\cite{koska2024tree-approach} & $\bigOh{8^n}$\\
        Georges \etal~\cite{georges2025pauli} & $\bigOh{n4^n}$
    \end{tblr}
    \caption{Comparison of runtime complexity of various Pauli decomposition algorithms.}
    \label{tab:previous-algorithms}
\end{table}

\begin{figure*}[t]
    \centering
    \begin{tikzpicture}[
        x=1cm,y=1cm,
        >=Latex,
        font=\captionsize,
        arrow/.style={-Latex, line width=0.65pt, shorten >=1pt, shorten <=1pt},
        looparrow/.style={-Latex, line width=0.55pt, dashed, shorten >=1pt, shorten <=1pt},
        panel/.style={
            draw,
            rounded corners=2.5pt,
            line width=0.55pt,
            align=center,
            inner sep=2.5pt
        },
        title/.style={font=\captionsize\bfseries, inner sep=0pt},
        note/.style={font=\captionsize, align=center, inner sep=0pt},
        atom/.style={
            circle,
            draw,
            line width=0.4pt,
            minimum size=4.0mm,
            inner sep=0pt,
            font=\captionsize
        },
        atomsmall/.style={
            circle,
            draw,
            line width=0.4pt,
            minimum size=3.2mm,
            inner sep=0pt,
            font=\captionsize
        },
        bucket/.style={
            draw,
            line width=0.4pt,
            minimum width=7.5mm,
            minimum height=7.0mm,
            inner sep=0pt
        }
    ]

        \node[
            panel,
            fill=gray!8,
            minimum width=4.2cm,
            minimum height=1.50cm
        ] (input) at (1.60,5.20) {};
        
        \node[title] at (2.75,5.65) {input};
        
        \node[font=\captionsize] at (0.30,5.12) {$M=$};
        
        \begin{scope}[shift={(0.88,4.74)}, x=0.13cm, y=0.13cm]
        
            \foreach \i in {0,...,5} {
                \foreach \j in {0,...,5} {
                    \fill[blue!3] (\i,\j) rectangle ++(1,1);
                }
            }
        
            \foreach \i/\j/\s in {
                0/5/10, 1/5/6,  2/5/34, 3/5/8,  4/5/20, 5/5/7,
                0/4/7,  1/4/26, 2/4/11, 3/4/35, 4/4/8,  5/4/18,
                0/3/23, 1/3/6,  2/3/14, 3/3/10, 4/3/39, 5/3/13,
                0/2/9,  1/2/34, 2/2/5,  3/2/24, 4/2/15, 5/2/30,
                0/1/16, 1/1/9,  2/1/28, 3/1/5,  4/1/22, 5/1/8,
                0/0/5,  1/0/18, 2/0/7,  3/0/31, 4/0/10, 5/0/26
            }{
                \fill[blue!\s] (\i,\j) rectangle ++(1,1);
            }
        
            \fill[yellow!45, opacity=0.50] (0,2) rectangle (6,3);
            \fill[yellow!45, opacity=0.50] (4,0) rectangle (5,6);
        
            \draw[step=1, white, line width=0.25pt] (0,0) grid (6,6);
            \draw[black, line width=0.45pt] (0,0) rectangle (6,6);
        
            \draw[-Latex, line width=0.55pt] (-1.05,0.6) -- (0,0.6);
            \node[font=\captionsize, anchor=east] at (-1.18,0.6) {$v$};
        
            \draw[-Latex, line width=0.55pt] (0.6,6.95) -- (0.6,6.0);
            \node[font=\captionsize, anchor=south] at (0.6,7.02) {$u$};
        \end{scope}
        
        \node[font=\captionsize, align=center] at (2.75,5.12)
            {sparse \\query access};
        
        \node[
            panel,
            fill=blue!6,
            minimum width=4.2cm,
            minimum height=1.45cm
        ] (xfib) at (1.60,3.35) {};
        
        \node[title] at (1.60,3.91) {determine $\mathcal{X}$};
        \node[note] at (1.60,3.66) {$x = u \oplus v$};
        
        \node[atom, fill=green!18] at (0.93,3.30) {$x_1$};
        \node[atom, fill=red!16]   at (1.60,3.30) {$x_2$};
        \node[atom, fill=green!18] at (2.27,3.30) {$x_3$};
        
        \node[note] at (1.60,2.82) {$\mathcal X=\{x:\exists z\}$};
        
        \node[
            panel,
            fill=blue!6,
            minimum width=4.2cm,
            minimum height=1.45cm
        ] (wht) at (1.60,1.50) {};
        
        \node[title] at (1.60,2.06) {shifts on fixed-$x$};
        
        \node[font=\captionsize, align=center] at (1.60,1.50)
            {$\frac{b_x(e_j)}{b_x(0^n)} \overset{?}{=} \pm 1$ for $j \in [n]$};
        
        \draw[arrow] (input.south) -- (xfib.north);
        \draw[arrow] (xfib.south) -- (wht.north);
        
        \node[
            panel,
            fill=green!8,
            minimum width=2.85cm,
            minimum height=1.45cm
        ] (single) at (6.55,3.55) {};
        
        \node[title] at (6.55,4.09) {unique $x$};
        
        \node[atom, fill=green!20] at (5.95,3.55) {$x$};
        \node[atom, fill=green!20] at (7.15,3.55) {$z$};
        
        \draw[arrow] (6.20,3.55) -- (6.90,3.55);
        
        \node[note] at (6.55,3.04) {decode $(z,\alpha)$};
        
        \node[
            panel,
            fill=red!7,
            minimum width=3.25cm,
            minimum height=1.70cm
        ] (deg) at (5.85,0.10) {};
        
        \node[title] at (5.85,0.76) {degenerate $x$};
        
        \node[atom, fill=red!16] at (4.98,0.12) {$z_1$};
        \node[atom, fill=red!16] at (5.56,0.12) {$z_2$};
        \node[atom, fill=red!16] at (6.14,0.12) {$z_3$};
        \node[atom, fill=red!16] at (6.72,0.12) {$z_4$};
        
        \node[note, text width=3.00cm] at (5.85,-0.55)
            {fixed $x$, several $z$'s};
        
        \node[
            panel,
            fill=purple!7,
            minimum width=3.75cm,
            minimum height=1.70cm
        ] (fold) at (9.75,0.10) {};
        
        \node[title] at (9.75,0.77) {random fold};
        \node[note] at (9.75,0.48) {hash $(S,R,w)$};
        
        \foreach \x in {8.75,9.75,10.75} {
            \node[bucket, fill=white] at (\x,-0.04) {};
        }
        
        \node[atomsmall, fill=red!16]   at (8.70,0.08) {$1$};
        \node[atomsmall, fill=red!16]   at (8.93,-0.10) {$2$};
        \node[atomsmall, fill=green!20] at (9.75,-0.01) {$3$};
        \node[atomsmall, fill=green!20] at (10.75,-0.01) {$4$};
        
        \node[note, text width=3.45cm] at (9.75,-0.57)
            {keep verified singletons};
        
        \node[
            panel,
            fill=yellow!12,
            minimum width=3.40cm,
            minimum height=1.70cm
        ] (coll) at (13.75,0.10) {};
        
        \node[title] at (13.75,0.76) {repeat, collect};
        
        \node[atom, fill=green!20] at (13.5,0.12) {$z_3$};
        \node[atom, fill=green!20] at (14,0.12) {$z_4$};
        
        \node[note] at (13.75,-0.55) {all singletons};
        
        \node[
            panel,
            fill=gray!10,
            minimum width=3.65cm,
            minimum height=1.45cm
        ] (out) at (13.75,3.55) {};
        
        \node[title] at (13.75,4.09) {output};
        
        \node[font=\captionsize, align=center] at (13.75,3.55)
            {Pauli decomposition:\\ $\mathcal{D} = \{P(x,z), \alpha_{x,z}\}$};
        
        \draw[arrow]
            (wht.east)
            to[out=35,in=270]
            node[left, note, pos=0.6, yshift=4pt] {$p_x = 1$}
            (single.south);
        
        \draw[arrow]
            (wht.south)
            to[out=-85,in=180]
            node[left, note, pos=0.50, yshift=-7pt] {$p_x > 1$}
            (deg.west);
        
        \draw[arrow] (single.east) -- (out.west);
        
        \draw[arrow] (deg.east) -- (fold.west);
        \draw[arrow] (fold.east) -- (coll.west);
        
        \draw[arrow] (coll.north) -- ++(0,1.55) -| (out.south);
        
        \draw[looparrow]
            (coll.south)
            .. controls (12.10,-2.25) and (10.15,-1.25) ..
            node[below, note, pos=0.4, yshift=-4pt] {fresh independent folds}
            (fold.south);

        \begin{pgfonlayer}{background}
        \node[
            draw=red!55!black,
            dashed,
            rounded corners=3pt,
            inner xsep=8pt,
            inner ysep=8pt,
            fit=(deg)(fold)(coll)
        ] {};
        \end{pgfonlayer}
        
        \node[font=\captionsize\bfseries, red!55!black] at (9.35,1.43)
            {degenerate $x$ decoder};
    \end{tikzpicture}

    \caption{Schematic of the Pauli decomposition algorithm (\cref{alg:full-algo}). Assuming query access to the matrix $M$ as described in~\cref{def:query-access-model}, the algorithm makes a series of queries that determine the set of active $x$ bit strings and stores them in $\mathcal{X}$ (see the beginning the~\cref{sec:pauli-decomposition-algorithm}). Then, for each $x \in \mathcal{X}$, the algorithm determines if $x$ is unique and, if so, decodes it by finding its associated $z$ bit string and coefficient $\alpha$ to construct the full term in the Pauli decomposition (\cref{subsec:unique-x-bit-strings,alg:unique-x-bit-strings}). If $x$ is degenerate, the algorithm passes $x$ to the degenerate $x$ decoder (\cref{alg:determine-degenerate-x-pauli-string}), which uses the random folding procedure described in~\cref{subsec:degenerate-x-algorithm} to determine the $z$ bit strings and coefficients associated with each degenerate $x$. At the end, everything is put together to build the full Pauli decomposition of $M$.}
    \label{fig:algorithm-schematic}
\end{figure*}
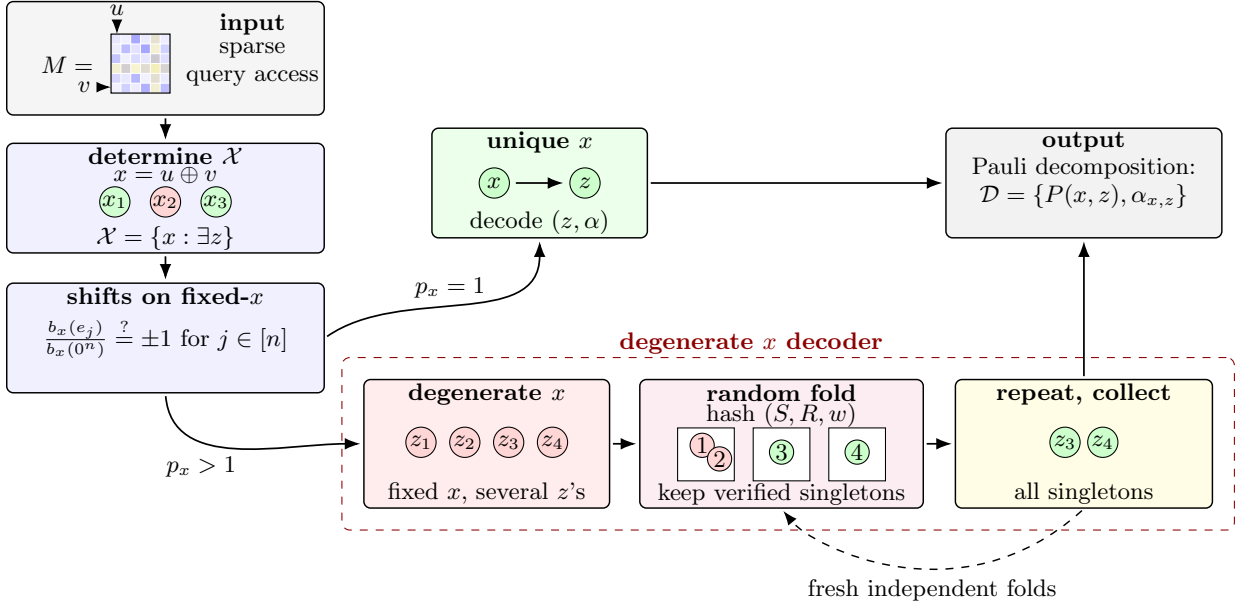

In this paper, we study a promise version of the Pauli decomposition problem. In particular, we assume that the input is a $2^n \times 2^n$ matrix that has support on $k = \poly(n)$ elements of the Pauli basis and that the matrix is specified via a classical sparse query model. Under this model, we present a randomized classical algorithm that determines the Pauli decomposition with success probability at least $1-\delta$ for a user-chosen $\delta \in (0,1)$, with both query and runtime complexity scaling polynomially in $n$, $k$, and $\log(1/\delta)$. The above-mentioned existing algorithms do not automatically become polynomial-time under the promise that the Pauli support has size $k=\poly(n)$ because there are still exponentially many nonzero entries and the algorithms require the full matrix.

Our algorithm first identifies the distinct Pauli-$X$ parts present in the decomposition, resolves the Pauli strings that have a unique $X$ part, and then handles the non-unique (or \emph{degenerate}) parts using sparse Walsh-Hadamard decoding. The Walsh-Hadamard transform also plays a key role in the dense Pauli decomposition algorithm of Georges \etal~\cite{georges2025pauli}, which computes the full set of Pauli coefficients for an arbitrary input matrix in time $\bigOh{n4^n}$. In contrast, in our setting, we can avoid computing all $4^n$ coefficients by using the promised sparsity, a randomized support discovery procedure, and sparse Walsh-Hadamard techniques to identify only the $\poly(n)$ nonzero terms efficiently. Thus, while the unrestricted Pauli decomposition problem remains exponentially costly in the worst case, we show that it becomes efficiently solvable for matrices with polynomial Pauli-support.

The rest of this article is organized as follows. In~\cref{sec:prelims}, we review the formalism and notation for the Pauli group and give a brief overview of basic Walsh-Hadamard analysis. In~\cref{sec:pauli-decomposition-algorithm}, we define the problem and give an overview of our decomposition algorithm. We benchmark our algorithm's performance on numerical examples in~\cref{sec:numerics}. Finally, we offer some concluding remarks in~\cref{sec:discussion} to put our algorithm in context. We relegate detailed analysis of the correctness, error, query complexity, and runtime complexity of our algorithm to~\cref{app-sec:correctness-error-analysis,app-sec:query-complexity,app-sec:time-complexity,app-sec:overall-guarantee}.

\section{Preliminaries}\label{sec:prelims}
We start by providing some background on the Pauli group~\cite{aguilar2024full} and introducing the notation that we use in this work. We also give a brief overview of Walsh-Hadamard analysis over the Boolean cube, where we present several results that we will use to develop our algorithm.

\subsection{The Pauli group}\label{subsec:pauli-group-background}
Throughout, we make use of the notation $[n] := \{1, 2, \ldots, n\}$. Let the Pauli matrices be denoted by
\begin{align}
    I &= \idmat, \quad X = \PauliX,\\
    Y &= \PauliY, \quad Z = \PauliZ.
\end{align}
With this, we define the $n$-qubit Pauli group:

\begin{definition}[Pauli group]\label{def:pauli-group}
    Let $n \in \NN$. The degree-$n$ \emph{Pauli group}, denoted $\mathcal{P}_n$, is the set $\mathcal{P}_n \coloneqq \lc \omega\bigotimes_{j=1}^{n}{P_j} : \omega \in \{\pm 1, \pm i\}, P_j \in \{I,X,Y,Z\}\ \forall j \in [n] \rc$, where the cardinality of the group is $\abs{\mathcal{P}_n} = 4^{n+1}$.
\end{definition}

In this paper, we do not need to keep track of all four phases $\omega \in \{\pm 1, \pm i\}$ separately. This motivates the following definition of \emph{Pauli strings}:

\begin{definition}[Pauli string]\label{def:pauli-string}
    A degree-$n$ \emph{Pauli string} is an $n$-fold tensor product of the single-qubit Pauli matrices:
    \begin{equation}
        P = \bigotimes_{j=1}^{n}{P_j},
    \end{equation}
    where $P_j \in \{I,X,Y,Z\}$ for all $j \in [n]$.
\end{definition}

We denote by $\hat{\mathcal{P}}_n$ the set of all such degree-$n$ Pauli strings and thus $\abs{\hat{\mathcal{P}}_n} = 4^n$. Therefore, while $\mathcal{P}_n$ is the full Pauli group, the set $\hat{\mathcal{P}}_n$ is the relevant object for the purposes of finding the Pauli decomposition of a matrix and is thus what we work with in this paper. $\hat{\mathcal{P}}_n$ forms a complete basis for the set of $2^n \times 2^n$ complex-valued matrices~\cite[Thm. 1]{koska2024tree-approach}. That is, if $M \in \CC^{2^n \times 2^n}$, then $M$ can be written as
\begin{align}
    M = \defsum{i=1}{k}{\alpha^{(i)} P^{(i)}},\label{eqn:M-linear-combination}
\end{align}
where we denote the $i\numth$ Pauli string in the sum by $P^{(i)} \in \hat{\mathcal{P}}_n$ and its coefficient as $\alpha^{(i)} \in \CC$, where $i \in [k]$. In the worst case, $M$ has support on all $4^n$ Pauli strings. We note that if the matrix to be decomposed into the Pauli basis is Hermitian, then the coefficients are restricted to be real: $M = M^\dagger \iff \alpha^{(i)} \in \RR$ for all $i \in [k]$.

The above representation is conceptually simple, but for the algorithm that we propose, it is more convenient to label each Pauli string by binary strings. In particular, we make use of the \emph{binary symplectic representation} of Pauli strings:

\begin{definition}[Binary symplectic form]\label{def:binary-symplectic-form}
    Let $x,z \in \{0,1\}^n$ be $n$-bit strings. The Pauli string labeled $(x \mid z)$ is
    \begin{align}
        P(x,z) &\coloneqq (-i)^{\abs{x \wedge z}}\bigotimes_{j=1}^{n}{Z^{z_j}X^{x_j}},
    \end{align}
    where
    \begin{equation}
        \abs{x \wedge z } \coloneqq \defsum{j=1}{n}{x_j z_j}
    \end{equation}
    counts the number of overlapping 1's in $x$ and $z$.
\end{definition}

Thus, the pair $(x_j, z_j)$ determines the single-qubit Pauli operator acting on the $j\numth$ qubit:
\begin{equation}
    \begin{aligned}
        (0,0) &\mapsto I,\quad (1,0) \mapsto X,\\
        (0,1) &\mapsto Z,\quad (1,1) \mapsto Y.
    \end{aligned}
\end{equation}
The prefactor $(-i)^{\abs{x \wedge z}}$ ensures the correct imaginary factors in the Pauli-$Y$ matrices, which appear when $x_j = z_j = 1$. The bit strings $x$ and $z$ thus give a complete description of each Pauli string. When we write a matrix $M$ as a linear combination of Pauli strings, as in~\cref{eqn:M-linear-combination}, we denote the binary symplectic labels of the $i\numth$ Pauli string by $x^{(i)}$ and $z^{(i)}$, where $x^{(i)}$ is the $X$-part and $z^{(i)}$ is the $Z$-part, that is, $(x^{(i)} \mid z^{(i)})$. This notation will be central in~\cref{sec:pauli-decomposition-algorithm} when we develop the algorithm.

We conclude this subsection with two useful properties of Pauli strings. First, the following lemma formalizes the relation between the rows and columns of Pauli strings:

\begin{lemma}[Row and column properties of Pauli strings]\label{lemma:row-column-pauli-strings}
    Let $P(x,z)$ be a degree-$n$ Pauli string written in binary symplectic form as in~\cref{def:binary-symplectic-form},
    \begin{align}
        P(x,z) = (-i)^{\abs{x \wedge z}}\bigotimes_{j=1}^{n}{Z^{z_j}X^{x_j}},
    \end{align}
    where $x,z \in \{0,1\}^n$. Assuming zero-indexing, denote the row indices of $P(x,z)$ by $v$ and its column indices by $u$, where $v, u \in \{0,1\}^n$ are $n$-bit strings, and note that each row and column of $P(x,z)$ has a single nonzero element. Then, for the nonzero matrix element $P(x,z)_{v,u}$ in row $v$ and column $u$, the following is true:
    \begin{align}
        x \oplus u \oplus v = 0^n,
    \end{align}
    where $\oplus$ denotes bit-wise addition modulo 2 and $0^n = 0 \ldots 0$ is the all-zero string.
\end{lemma}

\begin{proof}
    Looking at the nonzero element $P(x,z)_{v,u}$, we require each of $(Z^{z_j}X^{x_j})_{v_j u_j} \neq 0$ for all $j \in [n]$. Then, looking at each term in $P(x,z)$ individually, if $x_j = 0$, then $Z^{z_j}X^{x_j} = Z^{z_j}$ is diagonal and so the nonzero elements are exactly specified by $u_j = v_j$. Since $u_j, v_j \in \{0,1\}$, this is equivalent to $u_j \oplus v_j = 0$. Likewise, if $x_j = 1$, then $Z^{z_j}X^{x_j}$ is off-diagonal and so the nonzero elements are exactly specified by $u_j \neq v_j$, which is equivalent to $u_j \oplus v_j = 1$. Putting this together, the matrix element is nonzero if and only if $x_j \oplus v_j \oplus u_j = 0$ and this must hold for all $j \in [n]$, so we conclude that $x \oplus v \oplus u = 0^n$. Observe that because $x$ is fixed, for any given $u$ (or $v$), there is exactly one $v$ (or $u$) that "matches" and ensures that $x \oplus v \oplus u = 0^n$, so there is exactly one nonzero element in each row and column.
\end{proof}

Another useful relation is the value of the nonzero matrix element:

\begin{lemma}[Pauli string matrix element]\label{lemma:pauli-string-matrix-element}
    Let $P(x,z)$ be a degree-$n$ Pauli string described by $x, z \in \{0,1\}^n$. Then, the matrix element in row $v$ and column $u$ is given by
    \begin{align}
        P(x,z)_{v,u} = 
        \begin{cases}
            i^{\abs{x \wedge z}}(-1)^{z \cdot u}, & v = x \oplus u\\
            0 & \text{otherwise}.
        \end{cases}
    \end{align}
    where $v, u \in \{0,1\}^n$ and $z \cdot u \coloneqq z_1 u_1 + \cdots + z_n u_n \pmod{2}$ denotes the bitwise inner product modulo 2.
\end{lemma}

\begin{proof}
    We compute the matrix element one qubit at a time. For each $j \in [n]$, we have
    \begin{equation}
        X^{x_j}\ket{u_j} = \ket{u_j \oplus x_j},
    \end{equation}
    so
    \begin{equation}
        Z^{z_j} X^{x_j}\ket{u_j} = (-1)^{z_j(u_j \oplus x_j)}\ket{u_j \oplus x_j}.
    \end{equation}
    Thus, we have
    \begin{equation}
        (Z^{z_j} X^{x_j})_{v_j,u_j} = \delta_{v_j,u_j \oplus x_j}(-1)^{z_j(u_j \oplus x_j)}.
    \end{equation}
    Taking the tensor product across all $j \in [n]$ for the whole Pauli string, we have
    \begin{equation}
        P(x,z)_{v,u} = (-i)^{\abs{x \wedge z}}\defprod{j=1}{n}{\delta_{v_j, u_j \oplus x_j}(-1)^{z_j(u_j \oplus x_j)}}.
    \end{equation}
    By~\cref{lemma:row-column-pauli-strings}, the matrix element $P(x,z)_{v,u} \neq 0$ if and only if
    \begin{equation}
        v_j = u_j \oplus x_j\ \forall j \in [n],
    \end{equation}
    which is equivalent to saying that $v = x \oplus u$ and so we are left with
    \begin{equation}
        P(x,z)_{v,u} = (-i)^{\abs{x \wedge z}}\defprod{j=1}{n}{(-1)^{z_j(u_j \oplus x_j)}}.
    \end{equation}
    Using the fact that
    \begin{equation}
        z_j(u_j \oplus x_j) \equiv z_ju_j + z_jx_j \pmod{2},
    \end{equation}
    we have
    \begin{align}
        \defprod{j=1}{n}{(-1)^{z_j(u_j \oplus x_j)}} &= (-1)^{z \cdot u}(-1)^{\defsum{j=1}{n}{z_jx_j}}\\
        &= (-1)^{z \cdot u}(-1)^{\abs{x \wedge z}}.
    \end{align}
    Thus, putting everything together, we have
    \begin{align}
        P(x,z)_{v,u} &= (-i)^{\abs{x \wedge z}}(-1)^{\abs{x \wedge z}}(-1)^{z \cdot u}\\
        &= i^{\abs{x \wedge z}}(-1)^{z \cdot u},
    \end{align}
    which proves the claim.
\end{proof}

\subsection{Walsh-Hadamard analysis}\label{subsec:walsh-hadamard-analysis}
As we will see when we present the Pauli decomposition algorithm in~\cref{sec:pauli-decomposition-algorithm}, we will make use of basic Walsh-Hadamard analysis on the Boolean cube $\{0,1\}^n$~\cite{odonnell2014analysis,kushilevitz1993learning,scheibler2015fast}. We work over the vector space $\FF_2^n$, which we identify with the set of $n$-bit strings $\{0,1\}^n$. We define the \emph{Walsh character} as the $\pm 1$-valued function
\begin{equation}
    \chi_z(u) \coloneqq (-1)^{z \cdot u}.\label{eqn:walsh-character}
\end{equation}
The Walsh characters $\chi_z(u) = (-1)^{z \cdot u}$ are the Boolean-cube equivalents of Fourier modes. Thus, the Walsh-Hadamard transform expands a function $f$ as a linear combination of these Walsh character modes. In our setting, the variable $u$ will be a column index for the matrix $M$ and the frequency label $z$ is the Pauli-$Z$ part of a Pauli string. With this, we define the \emph{Walsh-Hadamard transform}:

\begin{definition}[Walsh-Hadamard transform]\label{def:walsh-hadamard-transform}
    Let $f : \{0,1\}^n \to \CC$. The \emph{Walsh-Hadamard transform} of $f$ is the function $\hat{f} : \{0,1\}^n \to \CC$ such that
    \begin{equation}
        f(u) = \defsum{z \in \{0,1\}^n}{}{\hat{f}(z)\chi_z(u)},
    \end{equation}
    where $\chi_z(u)$ is the Walsh character given in~\cref{eqn:walsh-character}.
\end{definition}

In this work, we will find that the Walsh-Hadamard transform matches the sign pattern produced by Pauli strings, which is a useful tool in our algorithm. In particular, the query function $b_x(u)$ that we define in~\cref{sec:pauli-decomposition-algorithm} will have exactly this form, with support indexed by the $z$ bit strings for a fixed $x$ bit string present in the decomposition. Next, we define the \emph{Walsh support}:

\begin{definition}[Walsh support]\label{def:walsh-support}
    Let $f : \{0,1\}^n \to \CC$ with Walsh-Hadamard transform $\hat{f} : \{0,1\}^n \to \CC$. The \emph{Walsh support} of $f$ is
    \begin{equation}
        \mathrm{supp}(\hat{f}) \coloneqq \{z \in \{0,1\}^n : \hat{f}(z) \neq 0\}.
    \end{equation}
    We say that $f$ is $s$-sparse in the Walsh domain if $\abs{\mathrm{supp}(\hat{f})} = s$.
\end{definition}

Thus, if a function on the Boolean cube can be written as the sum of only a few Walsh characters, then we say that its Walsh spectrum is sparse. Recovering (i.e., finding) this sparse spectrum is exactly the classical decoding problem that we use to deal with \emph{degenerate} $x$ bit strings (i.e., two or more Pauli strings described by the same $x$ bit string but that have different $Z$-parts) in our Pauli decomposition algorithm. Classical Fourier-based algorithms frequently exploit sparsity to greatly enhance efficiency~\cite{kushilevitz1993learning,scheibler2015fast}.

Using these notions, we can use \emph{shifts} to reveal information about hidden spectra. Let $f : \{0,1\}^n \to \CC$ be a function with Walsh-Hadamard expansion
\begin{equation}
    f(u) = \defsum{z \in \{0,1\}^n}{}{\hat{f}(z)(-1)^{z \cdot u}},
\end{equation}
and let $w \in \{0,1\}^n$ be a \emph{shift}. Define the shifted function
\begin{equation}
    f^{(w)}(u) \coloneqq f(u \oplus w).
\end{equation}
Then,
\begin{equation}
    f^{(w)}(u) = \defsum{z \in \{0,1\}^n}{}{\hat{f}(z)(-1)^{z \cdot w}(-1)^{z \cdot u}},
\end{equation}
which is a result of a simple application of the linearity of the inner product over $\FF_2$. This relation is key to the decoding procedure that we use to break the degeneracy of non-unique $x$ bit strings and recover the present $z$ bit strings. Once a single spectral component $z$ has been isolated, shifts reveal the signs $(-1)^{z \cdot w}$. In particular, taking $w = e_j$, where $e_j$ are standard basis vectors with a 1 in the $j\numth$ position and 0's elsewhere, reveals the individual bits of $z$, allowing us to build up each $z$ bit-by-bit. However, applying this idea directly to an $n$-bit Walsh-Hadamard transform would require evaluating the function on all $2^n$ inputs. Instead, we will find a \emph{folding} operation useful, which allows us to hash $n$-bit spectral indices into a smaller number of $m$-bit bins, with $m \ll n$:

\begin{lemma}[Shifted folding identity]\label{lemma:folding-identity}
    Let $f : \{0,1\}^n \to \CC$ have a Walsh-Hadamard expansion
    \begin{equation}
        f(u) = \defsum{z\in\{0,1\}^n}{}{\hat{f}(z)(-1)^{z \cdot u}}.
    \end{equation}
    Let $R \in \{0,1\}^{m \times n}$ be a binary matrix (i.e., a matrix whose elements are from $\ZZ/2\ZZ \cong \FF_2$) and define the \emph{shifted folded query function} on shift $w \in \{0,1\}^n$ as
    \begin{equation}
        y^{(w)}(t) \coloneqq f(R^\intercal t \oplus w),\quad t \in \{0,1\}^m.
    \end{equation}
    Then, $y^{(w)}$ has Walsh-Hadamard expansion
    \begin{equation}
        y^{(w)}(t) = \defsum{s\in\{0,1\}^m}{}{g^{(w)}(s)(-1)^{s \cdot t}},
    \end{equation}
    where
    \begin{equation}
        g^{(w)}(s) = \defsum{z:Rz=s}{}{\hat{f}(z)(-1)^{z \cdot w}}
    \end{equation}
    is the \emph{shifted folded spectrum}. In the unshifted case where $w = 0^n$, this reduces to
    \begin{equation}
        g(s) = \defsum{z:Rz=s}{}{\hat{f}(z)}.
    \end{equation}
    Thus, $R$ folds the original $n$-bit Walsh spectrum into $2^m$ \emph{bins} indexed by $s \in \{0,1\}^m$, and each bin contains the total contribution of all $z$ such that $Rz = s$.
\end{lemma}

\begin{proof}
    If we expand $f(R^\intercal t \oplus w)$ in terms of its Walsh characters, we have
    \begin{align}
        y^{(w)}(t) &= \defsum{z\in \{0,1\}^n}{}{\hat{f}(z)(-1)^{z \cdot (R^\intercal t \oplus w)}}\\
        &= \defsum{z\in\{0,1\}^n}{}{\hat{f}(z)(-1)^{z \cdot w}(-1)^{(Rz) \cdot t}}.
    \end{align}
    Grouping together terms $s = Rz$, we have
    \begin{equation}
        y^{(w)}(t) = \defsum{s\in\{0,1\}^m}{}{\parens*{\defsum{z:Rz=s}{}{\hat{f}(z)(-1)^{z \cdot w}}}(-1)^{s \cdot t}}.
    \end{equation}
    If we compare this with the Walsh-Hadamard expansion of $y^{(w)}$ on $\{0,1\}^m$, the term in front of $(-1)^{s \cdot t}$ is exactly
    \begin{equation}
        g^{(w)}(s) = \defsum{z:Rz=s}{}{\hat{f}(z)(-1)^{z \cdot w}},
    \end{equation}
    proving the claim.
\end{proof}

This result allows us to randomly group the $z$ support into a small number of bins. A bin $s$ is called a \emph{singleton} bin if exactly one nonzero coefficient $\hat{f}(z)$ satisfies $Rz = s$. In that case, $g(s) = \hat{f}(z)$ and the shifted spectra satisfy
\begin{equation}
    \frac{g^{(w)}(s)}{g(s)} = (-1)^{z \cdot w}.
\end{equation}
Taking $w = e_j$ thus reveals the $j\numth$ bit of $z$, as mentioned above. On the other hand, if several elements collide and end up in the same bin, these ratios need not correspond to a single $z$, and further processing is needed. This, along with the result in~\cref{lemma:folding-identity}, allows us to reduce an $n$-bit sparse recovery problem to a small $m$-bit transform problem, where if we choose $m \ll n$, the procedure becomes efficient.

Finally, we invoke an uncertainty principle result, which we will use in the development of the algorithm to check our candidate $z$ bit strings:

\begin{fact}[Boolean uncertainty principle]\label{fact:boolean-uncertainty-principle}
    Let $f : \{0,1\}^n \to \CC$ be a nonzero function and let $\hat{f}$ denote its Walsh-Hadamard spectrum. Then,
    \begin{equation}
        \abs{\mathrm{supp}(f)} \cdot \abs{\mathrm{supp}(\hat{f})} \geq 2^n.
    \end{equation}
    In other words, if $\hat{f}$ has support size at most $s$, then a uniformly random input $u$ satisfies $f(u) \neq 0$ with probability at least $1/s$.
\end{fact}

This fact follows from the classical uncertainty inequality for finite Abelian groups~\cite{meshulam2006uncertainty}. Applied to $G = (\ZZ_2)^n$, this gives $\abs{\mathrm{supp}(f)} \cdot \abs{\mathrm{supp}(\hat{f})} \geq 2^n$. Intuitively,~\cref{fact:boolean-uncertainty-principle} states that a nonzero function on the Boolean cube cannot be both sparse in the input domain and sparse in the Walsh domain. Thus, if $\hat{f}$ has support size $s$, then $f$ is a linear combination of only $s$ Walsh characters. This can result in cancellations depending on the coefficients, but~\cref{fact:boolean-uncertainty-principle} quantifies how extensive these cancellations can be: unless the combination is identically zero, it must be nonzero on at least $1/s$ of the Boolean cube. Thus, a uniformly random input detects a nonzero $s$-sparse Walsh spectrum with probability at least $1/s$. This is the primary consequence of the uncertainty principle that we use in our algorithm and allows us to determine whether a proposed singleton $z$ bit string is erroneous with inverse-polynomial probability.

With these ideas in hand, we note that dealing with the degenerate $x$ bit strings in the Pauli decomposition can be viewed as a sparse Walsh-Hadamard decoding problem on the Boolean cube. That is, for a fixed $x$, we have to recover a sparse set of $z$ indices and corresponding coefficients using oracle access to a function $b_x(u)$, which is a sum of a limited number of Walsh characters. The remainder of our paper specializes this idea to the setting of decomposing a matrix into the Pauli basis.

\section{Pauli decomposition algorithm}\label{sec:pauli-decomposition-algorithm}
The most general Pauli decomposition problem takes as input an arbitrary $2^n \times 2^n$ complex-valued matrix and asks for its Pauli basis expansion. In this work, we consider a problem in which we are \emph{promised} that the input matrix $M$ has polynomially-bounded support on the Pauli basis. We formalize this problem as follows:

\begin{problem}[Pauli-sparse decomposition problem]\label{prob:main-problem}
    Let $n \in \NN$ and $M \in \CC^{2^n \times 2^n}$. $M$ is promised to be exactly equal to a linear combination of $k$ Pauli strings:
    \begin{align}
        M = \defsum{i=1}{k}{\alpha^{(i)} P^{(i)}},\label{eqn:input-matrix}
    \end{align}
    where $\alpha^{(i)} \in \CC$ and $P^{(i)} \in \hat{\mathcal{P}}_n$ for all $i \in [k]$. Give an algorithm that returns all Pauli strings $P^{(i)}$ appearing in the decomposition along with their coefficients $\alpha^{(i)}$.
\end{problem}

To solve~\cref{prob:main-problem}, we assume \emph{sparse classical query access} to $M$:

\begin{definition}[Sparse query access model]\label{def:query-access-model}
    Let $M \in \CC^{2^n \times 2^n}$ be of the form in~\cref{eqn:input-matrix}. Then, since every Pauli string has exactly one nonzero entry in each row and column, each row and column of $M$ has at most $k$ nonzero entries. In the \emph{sparse query access model}, given a row index $v \in \{0,1\}^n$, a row query returns the list $\{(u, M_{v,u}) : M_{v,u} \neq 0\}$. When accounting for the query complexity of the algorithm, each such row retrieval counts as a single query.
\end{definition}

It is sufficient to assume either sparse row-query access or sparse column-query access. For concreteness, we state the algorithm using row queries for discovering the active $x$ bit strings and evaluations of the query function, defined as
\begin{equation}
    b_x(u) \coloneqq M_{x \oplus u, u},
\end{equation}
in terms of columns $u \in \{0,1\}^n$. Each evaluation of $b_x(u)$ can be implemented from a row query by querying row $x \oplus u$ and checking whether column $u$ appears in the returned sparse row. Equivalently, one could implement the same procedure using only column queries by querying column $u$. We discuss settings in which such a setup could arise in~\cref{sec:discussion}.

We claim that, if $k = \poly(n)$ and we are given query access to $M$ in the manner defined in~\cref{def:query-access-model}, then there exists an efficient randomized classical algorithm that can efficiently decompose $M$ into the Pauli basis:

\begin{theorem}[Pauli decomposition for sparse matrices]\label{thm:pauli-decomposition}
    For every $\delta \in (0,1)$, there exists a randomized classical algorithm that, given query access to $M \in \CC^{2^n \times 2^n}$ according to~\cref{def:query-access-model} and given a known sparsity $k$, returns the Pauli decomposition of $M$ with probability $1-\delta$. The algorithm uses a number of queries and has a total runtime complexity that are polynomial in $n$, $k$, and $\log(1/\delta)$.
\end{theorem}

We note that the algorithm can fail in two different ways: it can fail to pass one of the several certification steps that we have throughout the algorithm (in which case the algorithm returns \texttt{Fail}), but it can also "silently" fail by returning an incorrect decomposition without knowing that this decomposition is incorrect. We discuss the failure mechanisms of the algorithm in more detail in~\cref{app-sec:correctness-error-analysis}.

We also note that the algorithm that we present does not require $k$ to be known exactly. It is sufficient to be given an explicit upper bound $k_{\max}$ on the support size, with $k_{\max} \geq k$. All parameters for sampling, certification, etc. in the algorithm can then be chosen using $k_{\max}$ in place of $k$, and the resulting query and runtime bounds are polynomial in $n$, $k_{\max}$, and $\log(1/\delta)$. However, for simplicity, we assume that $k$ is known and so we write our analysis in terms of $k$ throughout.

We now describe the algorithm. Note, the goal is to find every Pauli string in the decomposition of $M$, which is equivalent to finding the triples $(x^{(i)}, z^{(i)}, \alpha^{(i)})$ for all terms in the sum, when writing each Pauli string $P^{(i)}$ in binary symplectic form (see~\cref{def:binary-symplectic-form}). For a fixed $x \in \{0,1\}^n$, denote by $\mathcal{P}_x$ the set of Pauli strings whose Pauli-$X$-part is $x$, and define $p_x \coloneqq \abs{\mathcal{P}_x}$. Thus, $p_x$ is the number of Pauli strings sharing the same Pauli-$X$ part; if $p_x=1$, then we say that $x$ is \emph{unique}, otherwise, if $p_x \geq 2$, then we say that $x$ is \emph{degenerate}.

Furthermore, denote by $\alpha_{x,z}$ the coefficient of the Pauli string labeled by $(x \mid z)$, which is zero if no such Pauli string appears in the decomposition. Using~\cref{lemma:pauli-string-matrix-element} and the phase convention in~\cref{def:binary-symplectic-form}, we can absorb the fixed phase from the sites in which $x_j = z_j = 1$ into a modified coefficient
\begin{equation}
    \beta_x(z) \coloneqq i^{\abs{x \wedge z}}\alpha_{x,z},\label{eqn:alpha-beta-relation}
\end{equation}
where, recall,
\begin{equation}
    \abs{x \wedge z} = \defsum{j=1}{n}{x_jz_j}.
\end{equation}
It is the effect of $\beta_x(z)$ that we see when we query the matrix, and given $x$ and $z$, we can uncompute to find the true coefficient $\alpha_{x,z}$ in the decomposition; that is,
\begin{equation}
    \alpha_{x,z} = i^{-\abs{x \wedge z}}\beta_x(z).
\end{equation}
By~\cref{lemma:row-column-pauli-strings}, every Pauli string with Pauli-$X$ part described by $x$ contributes to the matrix element in row $v = x \oplus u$ and column $u$, and by~\cref{lemma:pauli-string-matrix-element}, the $u$-dependence of that contribution is exactly the Walsh character $(-1)^{z \cdot u}$, which we described in~\cref{subsec:walsh-hadamard-analysis}. Thus, the query function is equivalent to
\begin{equation}
    b_x(u) = \defsum{z\in\{0,1\}^n}{}{\beta_x(z)(-1)^{z \cdot u}}.\label{eqn:bx-query}
\end{equation}
In other words, for a fixed $x$, the function $b_x$ has Walsh-Hadamard expansion with coefficients $\beta_x$, and since $\beta_x$ has support size $p_x$, the function $b_x$ is $p_x$-sparse in the Walsh domain.

At this point, a straightforward approach to pursue is, for a fixed $x$, evaluate $b_x(u)$ for all $u \in \{0,1\}^n$, which requires $2^n$ queries, and then recover $\beta_x$ using a full Walsh-Hadamard transform in time $\bigOh{n2^n}$. This recovers all $z$-parts and the coefficients for the given $x$, but the runtime complexity scales exponentially in $n$, which is exactly what we wish to avoid. Our algorithm instead exploits the $p_x$-sparsity of $\beta_x$.

First, we determine which $x$ bit strings appear in the decomposition at all. By~\cref{lemma:row-column-pauli-strings}, if we query the first row, where $v = 0^n$, then every nonzero entry in that row occurs at a column $u$ such that
\begin{equation}
    x \oplus u \oplus v = 0^n,
\end{equation}
where $x$ is the Pauli-$X$ part of a Pauli string with a nonzero coefficient. If $v = 0^n$, this immediately becomes $x = u$. Thus, a single query to row $0^n$ returns the distinct $x$ bit strings that appear in the Pauli decomposition of $M$, and we denote this set by
\begin{equation}
    \mathcal{X} \subseteq \{0,1\}^n.
\end{equation}
The size of this set is $\abs{\mathcal{X}} \leq k$, which tells us the number of distinct $x$ bit strings, and it is bounded above by the number of Pauli strings, that we are promised are in the decomposition, $k$. While this single query to the first row would suffice for many cases, there are situations in which the contributions from the support elements with Pauli-$X$ part $x$ exactly cancel in the first row. To see an example of this, consider a simple one-qubit instance with $M = X - iY$. Both terms have a Pauli-$X$ part described by $x=1$, but in the first row ($v=0$), the contributions $X_{0,1}=1$ and $(-iY)_{0,1} = -1$ exactly cancel to 0, meaning our query would tell us there are no Pauli strings with $x=1$.

Instead, we can discover the $x$ bit strings using several row queries chosen uniformly at random. To see why this works, we fix an $x \in \{0,1\}^n$ that is truly present in the decomposition. Then, for a row index $v \in \{0,1\}^n$, we consider the matrix element at column $u = x \oplus v$. Using~\cref{eqn:bx-query}, this element is given by $M_{v,u} = b_x(x \oplus v)$. Since we choose $v$ uniformly at random, $u = x \oplus v$ is also uniformly random, and since $b_x$ has Walsh-support size $p_x$, by~\cref{fact:boolean-uncertainty-principle}, we have
\begin{align}
    \Pr_{v\sim\{0,1\}^n}[M_{v,u} \neq 0] &= \Pr_{u\sim\{0,1\}^n}[b_x(u) \neq 0]\\
    &\geq \frac{1}{p_x} \geq \frac{1}{k}.
\end{align}
Thus, each uniformly random row reveals any fixed active $x$ bit string with probability at least $1/k$. Using this, we have the following $x$ discovery procedure. First, we choose a failure budget $\eta_{\mathcal{X}} \in (0,1)$ and sample
\begin{equation}
    R_{\mathcal{X}} \coloneqq \ceil*{k\log\lp \frac{k}{\eta_{\mathcal{X}}} \rp}\label{eqn:Rx-value}
\end{equation}
independent uniformly sampled rows $v^{(1)}, \ldots, v^{(R_{\mathcal{X}})} \in \{0,1\}^n$. For each row $v^{(r)}$, take every returned nonzero entry's corresponding column index $u$ and add the bit string $x = u \oplus v^{(r)}$ to the set of $x$ bit strings $\mathcal{X}$. For any $x$ bit string that is truly in the decomposition, the probability that it is missed in all $R_{\mathcal{X}}$ rows that we query is at most
\begin{equation}
    \lp 1 - \frac{1}{k} \rp^{R_{\mathcal{X}}} \leq e^{-R_{\mathcal{X}}/k} \leq \frac{\eta_{\mathcal{X}}}{k},
\end{equation}
where we used the relation $1-a \leq e^{-a}$ and the expression for $R_{\mathcal{X}}$ in~\cref{eqn:Rx-value}. As there are at most $k$ active $x$ bit strings in the decomposition (i.e., the worst case for the count is when
all active x bit strings are unique), a union bound shows that every $x$ bit string is discovered with probability at least $1 - \eta_{\mathcal{X}}$.

Once we have the set $\mathcal{X}$ in hand, the remaining task is to recover, for each $x \in \mathcal{X}$, the corresponding $z$ bit string(s) and coefficient(s). We split this up into two sub-problems: dealing with the unique $x$ bit strings, for which $p_x = 1$, and dealing with the degenerate $x$ bit strings, for which $p_x \geq 2$. We summarize the full procedure in~\cref{alg:full-algo}. We start with decoding the unique $x$ bit strings.

\begin{figure}[tb]
    \begin{algorithm}[H]
        \caption{Pauli decomposition algorithm}\label{alg:full-algo}
        \begin{algorithmic}[1]
            \Procedure{\textsc{PauliDecomposition}}{$M,n,k,\delta$}
                \Require Query access to $M$, number of qubits $n$, sparsity $k$, failure probability $\delta$.
                \Ensure Set $\mathcal{D}$ of pairs $(P(x,z), \alpha_{x,z})$ or \texttt{Fail}
                \State Choose failure budgets
                \begin{equation}
                    \eta_{\mathcal{X}}=\frac{\delta}{3},\quad \eta_x=\frac{\delta}{3k},\quad \delta_x=\frac{\delta}{3k}
                \end{equation}
                \State Set $R_{\mathcal{X}} = \ceil*{k\log{\parens*{\frac{k}{\eta_{\mathcal{X}}}}}}$
                \State Initialize $\mathcal{X}, \mathcal{X}_{\mathrm{deg}}, \mathcal{D} \gets \emptyset, \emptyset, \emptyset$
                \For{$r = 1, \ldots, R_{\mathcal{X}}$}
                    \State sample $v^{(r)} \in \{0,1\}^n$ uniformly at random
                    \State query row $v^{(r)}$ of $M$
                    \State for each nonzero entry $(u,M_{v^{(r)},u})$, calculate $x = u \oplus v^{(r)}$, append to $\mathcal{X}$
                \EndFor
                \For{$x \in \mathcal{X}$}
                    \State run $\textsc{Unique}(x,\eta_x)$
                    \If{returns $(\mathtt{Unique},\hat{z},\beta)$}
                        \State set $\alpha_{x,z} = i^{-\abs{x \wedge \hat{z}}}\beta$, add $(P(x,\hat{z}), \alpha_{x,\hat{z}})$ to $\mathcal{D}$
                    \Else
                        \State add $x$ to $\mathcal{X}_{\mathrm{deg}}$
                    \EndIf
                \EndFor
                \For{$x \in \mathcal{X}_{\mathrm{deg}}$}
                    \State run $\textsc{Degenerate}(x, p^{\max}_x, \delta_x)$
                    \If{\texttt{Fail}}
                    \State \textbf{return} \texttt{Fail}
                    \EndIf
                    \For{each $(z,\beta_x(z))$ returned}
                        \State $\alpha_{x,z} = i^{-\abs{x \wedge z}}\beta_x(z)$
                        \State add $(P(x,z), \alpha_{x,z})$ to $\mathcal{D}$
                    \EndFor
                \EndFor
                \State \textbf{return} $\mathcal{D}$
            \EndProcedure
        \end{algorithmic}
    \end{algorithm}
\end{figure}

\subsection{Determining and decoding unique \texorpdfstring{$x$}{x} bit strings}\label{subsec:unique-x-bit-strings}
\begin{figure}[tb]
    \begin{algorithm}[H]
        \caption{Decode unique $x$ bit strings}\label{alg:unique-x-bit-strings}
        \begin{algorithmic}[1]
            \Procedure{Unique}{$x$, $\eta_x$}
                \Require Fixed $x \in \mathcal{X}$, failure budget $\eta_x$
                \Ensure $(\mathtt{Unique}, \hat{z}, \beta)$ or \texttt{Degenerate}
                \State Set $L_{\mathrm{cert}} \gets \ceil*{(k+1)\log(1/\eta_x)}$
                \State Evaluate $\beta = b_x(0^n)$
                \If{$\beta=0$}
                    \State \textbf{return} \texttt{Degenerate}
                \EndIf
                \State Initialize $\hat{z} \gets 0^n$
                \For{$j = 1, \ldots, n$}
                    \State evaluate $b_x(e_j)$
                    \State compute $r_j = \frac{b_x(e_j)}{\beta}$
                    \If{$r_j = +1$} set $\hat{z}_j=0$
                    \ElsIf{$r_j=-1$} set $\hat{z}_j=1$
                    \Else \textbf{ return} \texttt{Degenerate}
                    \EndIf
                \EndFor
                \For{$\ell = 1, \ldots, L_{\mathrm{cert}}$}
                    \State uniformly sample $u^{(\ell)} \in \{0,1\}^n$
                    \If{$b_x(u^{(\ell)}) \neq \beta \cdot (-1)^{\hat{z} \cdot u^{(\ell)}}$}
                        \State \textbf{return} \texttt{Degenerate}
                    \EndIf
                \EndFor
                \State \textbf{return} $(\mathtt{Unique}, \hat{z}, \beta)$
            \EndProcedure
        \end{algorithmic}
    \end{algorithm}
\end{figure}

With the set of $x$ bit strings $\mathcal{X}$, we now fix an $x \in \mathcal{X}$. Recall that
\begin{equation}
    b_x(u) = \defsum{z \in \{0,1\}^n}{}{\beta_x(z)(-1)^{z \cdot u}},
\end{equation}
where, recall, evaluating $b_x(u)$ is implemented by querying column $u$ and reading off the entry in row $x \oplus u$, which gives us matrix element $M_{x \oplus u, u}$. If $x$ is unique (i.e., $p_x=1$), then there exists a bit string $z^\star \in \{0,1\}^n$ such that the sum collapses to
\begin{equation}
    b_x(u) = \beta_x(z^\star)(-1)^{z^\star \cdot u}
\end{equation}
for all $u \in \{0,1\}^n$. That is, for a unique $x$, the function $b_x$ is a single Walsh character, and both the coefficient and the $z$ bit string can be recovered from a small number of additional queries. In particular, if we take $u = 0^n$ and evaluate $b_x(0^n)$, then we have
\begin{equation}
    b_x(0^n) = \beta_x(z^\star),
\end{equation}
which gives us the coefficient up to a phase, allowing us to exactly recover the true coefficient $\alpha_{x,z}$ using the relationship defined in~\cref{eqn:alpha-beta-relation}. Furthermore, if we evaluate $b_x(u)$ at the standard basis vectors $e_j$, then we can take the ratio
\begin{equation}
    \frac{b_x(e_j)}{b_x(0^n)} = (-1)^{z^\star_j}
\end{equation}
for each $j \in [n]$. Thus, if $x$ is unique, then these $n$ queries that we make (one for each $e_j$) determine the bits of $z^\star$ one-by-one. That is, for each $j$, the ratio $b_x(e_j)/b_x(0^n)$ must be equal to $+1$ or $-1$, corresponding to $z^\star_j = 0$ or $z^\star_j = 1$, respectively. If any of these ratios fail to lie in $\{\pm 1\}$, then that $x$ bit string cannot be unique and we declare $x$ degenerate.

While these queries give a candidate unique $x$ bit string, they do not by themselves rule out the possibility that a degenerate $x$ bit string agrees with a single Walsh character on the queried points $0^n, e_1, \ldots, e_n$. To certify whether a candidate unique $x$ bit string is truly unique, after constructing the corresponding candidate $\hat{z}$ bit string and coefficient $\beta = b_x(0^n)$, we evaluate $b_x(u)$ on a small number $L_{\mathrm{cert}}$ (to be chosen below) of uniformly random inputs $u^{(1)}, \ldots, u^{(L_{\mathrm{cert}})}$ and test whether
\begin{equation}
    b_x(u^{(\ell)}) = \beta \cdot (-1)^{\hat{z} \cdot u^{(\ell)}}
\end{equation}
holds for all $\ell \in [L_{\mathrm{cert}}]$. If every certification query passes, we accept $x$ as unique and conclude that $\hat{z}$ is in fact the correct $Z$-part and that $\beta = \beta_x(\hat{z})$ is the correct coefficient. The original Pauli coefficient is then found using~\cref{eqn:alpha-beta-relation}, and the corresponding Pauli string is 
\begin{equation}
    P(x,\hat{z}) = (-i)^{\abs{x \wedge \hat{z}}}\bigotimes_{j=1}^{n}{Z^{\hat{z}_j}X^{x_j}}.
\end{equation}
If, on the other hand, any of the certification queries fail, then we conclude that $x$ is degenerate, and we decode it using the degenerate $x$ bit string algorithm that we develop in the next subsection.

The reason that this certification step is effective is that, if we have a candidate pair $(\hat{z}, \beta)$ that is incorrect, then the discrepancy
\begin{equation}
    h_x(u) \coloneqq b_x(u) - \beta(-1)^{\hat{z} \cdot u}
\end{equation}
is a nonzero Walsh-sparse function. Its Walsh support has size at most $p_x+1 \leq k+1$, so by~\cref{fact:boolean-uncertainty-principle}, a uniformly random query detects the error with probability at least $1/(k+1)$. Thus, if we choose
\begin{equation}
    L_{\mathrm{cert}} = \ceil{(k+1)\log(1/\eta_x)},
\end{equation}
then we ensure that an incorrect candidate pair $(\hat{z}, \beta)$ passes the checks with probability at most $\eta_x$. We prove this statement in~\cref{app-sec:correctness-error-analysis}. We summarize this whole procedure in~\cref{alg:unique-x-bit-strings}.

\subsection{Decoding degenerate \texorpdfstring{$x$}{x} bit strings}\label{subsec:degenerate-x-algorithm}
We now consider the \emph{degenerate} $x$ bit strings, that is, those $x \in \{0,1\}^n$ for which $p_x \geq 2$. If we fix an $x \in \mathcal{X}$, then recall from~\cref{eqn:bx-query} that
\begin{equation}
    b_x(u) = \defsum{z\in\{0,1\}^n}{}{\beta_x(z)(-1)^{z \cdot u}},
\end{equation}
where the Walsh support of the function $b_x$ is of size $p_x$. Thus, recovering the Pauli strings with Pauli-$X$ part described by the fixed $x \in \mathcal{X}$ is equivalent to recovering the support and coefficients of the sparse Walsh-Hadamard spectrum $\beta_x$. That is, the degenerate $x$ bit string problem is a sparse Walsh-Hadamard problem on the Boolean cube, the machinery for which we introduced in~\cref{subsec:walsh-hadamard-analysis}.

\begin{figure}[tb]
    \begin{algorithm}[H]
        \caption{Decode degenerate $x$ bit strings}\label{alg:determine-degenerate-x-pauli-string}
        \begin{algorithmic}[1]
            \footnotesize
            \Procedure{\textsc{Degenerate}}{$x$, $p^{\max}_x$, $\delta_x$}
                \Require Degenerate $x \in \mathcal{X}_{\mathrm{deg}}$, bound $p^{\max}_x \geq p_x$, failure budget $\delta_x$
                \Ensure Set $\mathcal{S} = \{(z,\beta_x(z))\}$ or \texttt{Fail}
                \State Choose $c > 1$
                \State $m \gets \ceil*{\log(cp^{\max}_x)}$
                \State $B \gets 2^m$
                \State $T \gets \ceil*{\log_c{\parens*{\frac{3p^{max}_x}{\delta_x}}}}$
                \State $L \gets \ceil*{(p^{\max}_x+1)\log{\parens*{\frac{3BT}{\delta_x}}}}$
                \State $q \gets \ceil*{p^{\max}_x\log{\parens*{\frac{3}{\delta_x}}}}$
                \State Initialize $\mathcal{S} \gets \emptyset$
                \For{$\tau = 1, \ldots, T$}
                    \State sample $R \in \{0,1\}^{m \times n}$ uniformly at random
                    \State compute unshifted spectrum $g = \mathrm{Fold}(x,\mathcal{S},R,0^n)$
                    \For{$j = 1, \ldots, n$}
                        \State compute $g^{(e_j)} = \mathrm{Fold}(x,\mathcal{S},R,e_j)$
                    \EndFor
                    \State sample random shifts $w^{(1)}, \ldots, w^{(L)} \in \{0,1\}^n$
                    \For{$\ell = 1, \ldots, L$}
                        \State compute $g(w^{(\ell)}) = \mathrm{Fold}(x,\mathcal{S},R,w^{(\ell)})$
                    \EndFor
                \For{$s \in \{0,1\}^m$ with $g(s) \neq 0$}
                    \State initialize $\hat{z} \gets 0^n$
                    \For{$j = 1, \ldots, n$}
                        \State compute $r_j = \frac{g^{(e_j)}(s)}{g(s)}$
                        \If{$r_j = +1$} set $\hat{z}_j = 0$
                        \ElsIf{$r_j = -1$} set $\hat{z}_j = 1$
                        \Else \ reject bin and move onto next one
                        \EndIf
                    \EndFor
                    \If{$R\hat{z} \neq s$} reject bin and move onto next one
                    \EndIf
                    \State accept bin if $\frac{g^{(w^{(\ell)})}(s)}{g(s)} = (-1)^{\hat{z} \cdot w^{(\ell)}}$ for all $\ell \in [L]$
                    \If{candidate passes and $\hat{z} \notin \mathcal{S}$}
                        \State add $(\hat{z}, g(s))$ to $\mathcal{S}$
                    \EndIf
                \EndFor
            \EndFor
                \State After $T$ rounds, uniformly sample $u^{(1)}, \ldots, u^{(q)} \in \{0,1\}^n$
                \If{$\Tilde{b}_x(u^{(\ell)}) = 0$ for all $\ell \in [q]$}
                    \State \textbf{return} $\mathcal{S}$
                \Else
                    \State \textbf{return} \texttt{Fail}
                \EndIf
            \EndProcedure
        \end{algorithmic}
    \end{algorithm}
\end{figure}

The intuition behind this goes as follows. For a fixed $x \in \{0,1\}^n$, the function $u \mapsto b_x(u)$ is a sparse sum of Boolean characters $(-1)^{z \cdot u}$. The problem of dealing with the degenerate $x$ bit strings is then just a sparse Walsh-Hadamard inversion problem over the Boolean cube. The main idea is that we \emph{fold} the $n$-bit spectrum down to an $m$-bit spectrum, where $m \ll n$. This coarse-graining process forces the unknown $z$ bit strings into $2^m$ bins, and if a given bin contains only one true support element, then the corresponding coefficient can be read off directly. Furthermore, shifts applied to the input variable $u$ produce known phase changes in the Walsh-Hadamard domain, and this allows us to recover the bits of the single $z$ bit string in the isolated bin. We repeat this random folding process and subtract the terms that we have already recovered until we recover all of the $z$ bit strings and their associated coefficients with high probability. In essence, this process allows us to probe the Walsh-Hadamard spectrum of the query function $b_x$, rather than evaluating the full $2^n$-length spectrum, and thus decode the Pauli-$Z$ support associated with a fixed Pauli-$X$ part $x$.

To turn this intuition into an algorithm, we choose the folding dimension and repetition parameters as follows. First, we note that  we do not need to know the degeneracy $p_x$ exactly, and it is unlikely we would know this \emph{a priori} anyways. It is instead sufficient to know an upper bound
\begin{equation}
    p_x^{\max} \geq p_x,
\end{equation}
and a trivial choice\footnote{In practice, we can adaptively change $p_x^{\max}$ based on how many different $x$ bit strings we find initially, how many we have decoded so far, etc., but for simplicity of analysis, the trivial choice suffices.} is $p_x^{\max} = k$, where we assume in the statement of~\cref{prob:main-problem} that $k$ is known. Then, we choose a constant $c > 1$ and take
\begin{equation}
    m \coloneqq \ceil{\log(cp_x^{\max})},\quad B \coloneqq 2^m,
\end{equation}
such that the folded spectrum ends up having $B = 2^m = \bigTheta{p_x^{\max}}$ bins. We need to repeat the folding process over several \emph{rounds}, where we take the number of rounds to be
\begin{equation}
    T \coloneqq \ceil*{\log_c(3p_x^{\max}/\delta_x)},
\end{equation}
where $\delta_x \in (0,1)$ is the acceptable probability of failing to decode a fixed $x$ bit string. We show in~\cref{app-sec:correctness-error-analysis} how we choose $\delta_x$. Furthermore, we need to make a number of checks to ensure that we have successfully isolated a candidate singleton support element. We take this number to be
\begin{equation}
    L \coloneqq \ceil{(p_x^{\max}+1)\log(3BT/\delta_x)}.\label{eqn:number-random-shifts}
\end{equation}
These are the parameters that go into~\cref{alg:determine-degenerate-x-pauli-string} and we derive them in~\cref{app-sec:correctness-error-analysis}.

Throughout the algorithm, we maintain a set $\mathcal{S}$ of pairs $(z,\beta_x(z))$ that we have successfully recovered for the fixed $x$, and we define the \emph{residual oracle} as
\begin{equation}
    \Tilde{b}_x(u) \coloneqq b_x(u) - \defsum{(z,\beta_x(z))\in\mathcal{S}}{}{\beta_x(z)(-1)^{z \cdot u}}.
\end{equation}
At the beginning of a round, $\Tilde{b}_x$ contains the contributions of the currently undiscovered support elements. To isolate these elements, we first sample a random binary matrix $R \in \{0,1\}^{m \times n}$, and for each $t \in \{0,1\}^m$, we query
\begin{equation}
    y(t) \coloneqq \Tilde{b}_x(R^\intercal t),
\end{equation}
which, by~\cref{lemma:folding-identity}, has Walsh-Hadamard spectrum
\begin{equation}
    g(s) \coloneqq \defsum{z:Rz=s}{}{\Tilde{\beta}_x(z)},\quad s \in \{0,1\}^m,
\end{equation}
where
\begin{equation}
    \Tilde{\beta}_x(z) \coloneqq \beta_x(z) - \defsum{(z',\beta)\in\mathcal{S}}{}{\beta\vb{1}_{z=z'}}
\end{equation}
is the residual Walsh-Hadamard spectrum (i.e., the spectrum remaining after subtracting off the already-found support elements). Thus, with $B = 2^m$ queries, we can recover all bin values $g(s)$ by computing the Walsh-Hadamard spectrum of $y$ on $\{0,1\}^m$, which can be done efficiently as $m$ is logarithmic in $n$. We denote this whole procedure by
\begin{equation}
    \mathrm{Fold}(x,\mathcal{S},R,w)(s) = \defsum{z:Rz=s}{}{\Tilde{\beta}_x(z)(-1)^{z \cdot w}}\label{eqn:fold-function}
\end{equation}
for convenience, which we use as a short-hand in~\cref{alg:determine-degenerate-x-pauli-string}. So far, we have only considered unshifted queries, that is, $w = 0^n$, but $\mathrm{Fold}(x, \mathcal{S}, R, w)$ allows for any $w \in \{0,1\}^n$.

\begin{figure*}[tb]
    \centering
    \includegraphics[scale=0.6]{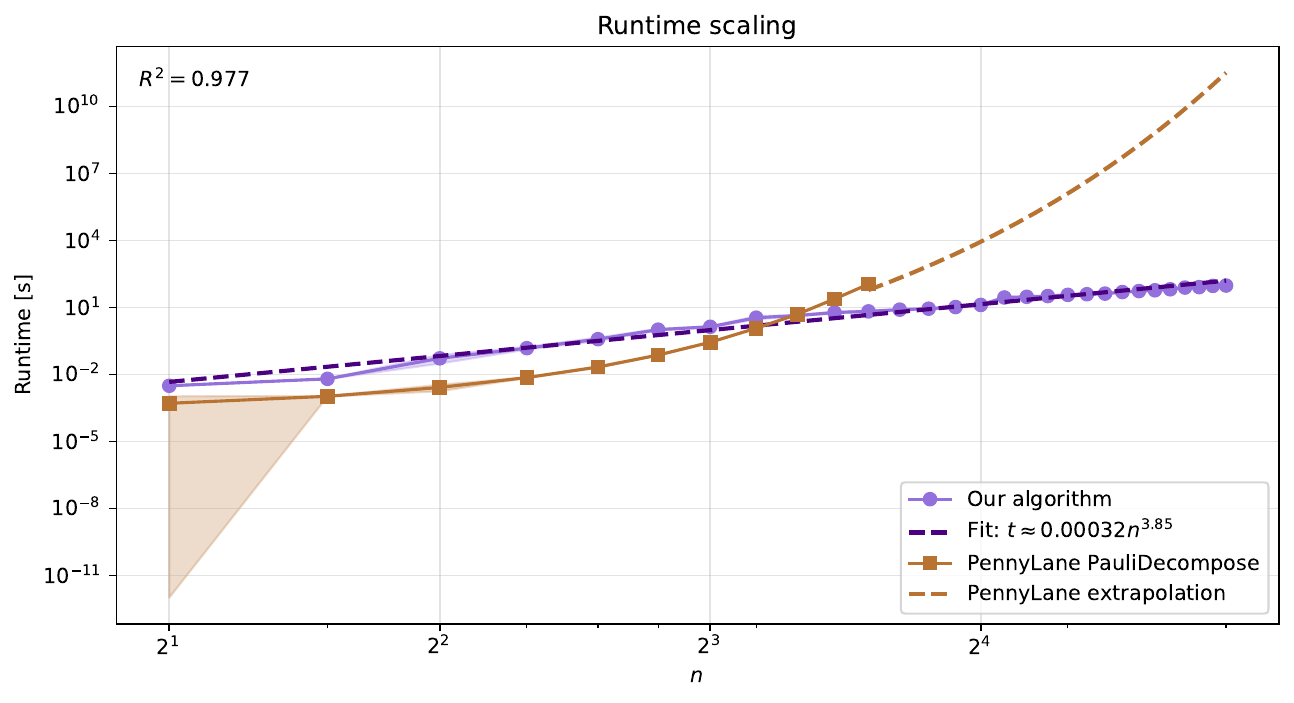}
    \caption{Runtime of our algorithm for 5 random instances each for values $n = 2, \ldots, 12$ and $k = 2n$ compared to PennyLane's \texttt{PauliDecompose} algorithm. Our algorithm is then applied to larger values of $n = 13, \ldots, 30$, while the PennyLane runtime was extrapolated using a fit. The observed scaling is consistent with the polynomial query and runtime bounds proved in~\cref{app-sec:correctness-error-analysis,app-sec:query-complexity,app-sec:time-complexity}.}
    \label{fig:runtime}
\end{figure*}

If a bin $s$ ends up containing exactly one support element $z^\star$ (i.e., a singleton), then we have
\begin{equation}
    g(s) = \Tilde{\beta}_x(z^\star) = \beta_x(z^\star),
\end{equation}
and so the coefficient can be read off immediately. Then, to find the bit string $z^\star$, we use a series of \emph{shifted queries}. That is, for any shift $w \in \{0,1\}^n$, we define
\begin{equation}
    y^{(w)}(t) \coloneqq \Tilde{b}_x(R^\intercal t \oplus w),
\end{equation}
and we denote by $g^{(w)}$ its folded spectrum. By~\cref{lemma:folding-identity}, we have
\begin{equation}
    g^{(w)}(s) = \defsum{z:Rz=s}{}{\Tilde{\beta}_x(z)(-1)^{z \cdot w}}.
\end{equation}
If $s$ is a singleton and contains the bit string $z^\star$, then
\begin{equation}
    \frac{g^{(w)}(s)}{g(s)} = (-1)^{z^\star \cdot w}.
\end{equation}
If we choose $w = e_j$, where $e_j$ is the $j\numth$ standard basis vector, then we have
\begin{equation}
    \frac{g^{(e_j)}(s)}{g(s)} = (-1)^{z^\star_j},
\end{equation}
which allows us to read off the full bit string $z^\star$ bit-by-bit.

We note that a collision bin (i.e., one that is not a singleton bin) can sometimes appear as a singleton bin on the basis-shift tests alone. Thus, we can check that each candidate singleton is truly a singleton by using a few additional random shifts $w^{(1)}, \ldots, w^{(L)} \in \{0,1\}^n$, and we accept a candidate $\hat{z}$ only if
\begin{equation}
    \frac{g^{(w^{(\ell)})}(s)}{g(s)} = (-1)^{\hat{z} \cdot w^{(\ell)}}
\end{equation}
holds true for all $\ell \in [L]$. If the candidate is not correct, then the discrepancy is a nonzero Walsh-sparse function with support size at most $p_x+1 \leq p_x^{\max}+1$, and, by~\cref{fact:boolean-uncertainty-principle}, each random shift detects the error with probability at least $1/(p_x^{\max}+1)$. With the choice of $L$ in~\cref{eqn:number-random-shifts}, a union bound over all $BT$ bins considered during the fixed-$x$ run implies that the probability that any false singleton is accepted is at most $\delta_x/3$. We prove this in~\cref{app-sec:correctness-error-analysis}.

Once we accept a term, we add $(\hat{z},g(s))$ to a set $\mathcal{S}$ and we \emph{peel} it from each subsequent query, and we repeat the above process with a fresh random fold. By choosing $B = \bigTheta{p_x^{\max}}$ bins and $T = \bigOh{\log(3p_x^{\max}/\delta_x)}$ rounds, we isolate every true support element for the fixed $x$ bit string with failure probability at most $\delta_x/3$. At the end, we perform a final randomized verification step by querying $\Tilde{b}_x(u)$ on $q=\bigOh{p_x^{\max}\log(3/\delta_x)}$ random inputs $u$. If every test returns zero (because we are subtracting all of the correctly recovered support elements), then we accept the recovered set $\mathcal{S}$; otherwise, we declare the procedure a failure. Combining the isolation guarantee, random-shift certification for false singletons, and the final residual check give a total failure probability of at most $\delta_x$, meaning we successfully decode a degenerate $x$ bit string with probability at least $1-\delta_x$. We summarize this whole procedure in~\cref{alg:determine-degenerate-x-pauli-string}.

In ~\cref{app-sec:correctness-error-analysis}, we prove that the whole algorithm, outlined in~\cref{alg:full-algo}, succeeds with probability at least $1-\delta$, where we choose the failure budgets such that
\begin{equation}
    \eta_{\mathcal{X}} + \defsum{x\in\mathcal{X}}{}{\eta_x} + \defsum{x\in\mathcal{X}_{\mathrm{deg}}}{}{\delta_x} \leq \delta.
\end{equation}
In particular, we choose
\begin{equation}
    \eta_{\mathcal{X}} = \frac{\delta}{3}, \quad
    \eta_x = \frac{\delta}{3k}, \quad
    \delta_x = \frac{\delta}{3k}.
\end{equation}
As we discuss further in~\cref{app-sec:correctness-error-analysis}, the algorithm can fail in two ways: it can fail and output \texttt{Fail} or it can produce an incorrect decomposition (i.e., failing without outputting \texttt{Fail}). Our error analysis encompasses both cases. On the complexity side, our algorithm requires only polynomially many queries in terms of $n$, $k$, and $\log(1/\delta)$ and has runtime that scales polynomially in $n$, $k$, and $\log(1/\delta)$. We prove both of these scalings in~\cref{app-sec:query-complexity,app-sec:time-complexity}, respectively, and present an overall guarantee for the algorithm in~\cref{app-sec:overall-guarantee}.

\section{Numerical experiments}\label{sec:numerics}
We now present numerical experiments deploying our algorithm on several random instances of~\cref{prob:main-problem}. These experiments incorporate the full algorithm as we developed it in~\cref{sec:pauli-decomposition-algorithm}: randomized $x$ bit string-support discovery, determining the uniqueness of the discovered $x$ bit strings, decoding the degenerate $x$ bit strings, and certification of results.

For each value of $n = 2, \ldots, 30$, we instantiate 5 random matrices $M$ of the form
\begin{equation}
    M = \defsum{i=1}{k}{\alpha^{(i)} P(x^{(i)},z^{(i)})},
\end{equation}
where the $\alpha^{(i)} \in \CC$ are chosen randomly and where the $P(x^{(i)},z^{(i)}) \in \hat{\mathcal{P}}_n$ are built by sampling uniformly at random $x^{(i)}, z^{(i)} \in \{0,1\}^n$, though we ensure\footnote{We ensure this by choosing some number $N_x$ of active $x$ bit strings we want in such a way that $N_x < k$. Since we have to assign $k$ $z$ bit strings to $N_x < k$ different $x$ bit strings, by the pigeonhole principle, some $x$ bit strings are forced to be degenerate.} that there is a subset of $x$ bit strings that are degenerate so as to make use of the subroutine in~\cref{subsec:degenerate-x-algorithm}. In our experiments, we choose the parameters $c=8$, $\delta=0.1$, and $k=2n$. All numerical experiments were run on an Intel Core i9-10900k CPU at 3.70 GHz with 10 physical cores
(20 logical threads) and 32 GB RAM, running Windows 10. The code was run
with Python v. 3.14.0, NumPy v. 2.3.5, PennyLane
v. 0.43.1, pandas v. 2.3.3, and matplotlib v. 3.10.7. The script used for these numerical experiments is publicly available at \url{https://github.com/dspencer2596/sparse-pauli-decomposition/releases/tag/v1.0-paper}.

In~\cref{fig:runtime}, we compare the runtime of our algorithm with the \texttt{PauliDecompose} algorithm in PennyLane~\cite{bergholm2018pennylane}, using values of $n = 2, \ldots, 12$ and building the matrix explicitly in memory. We pass the matrix, stored in classical memory in a sparse data structure (such that we only have to keep track of $\bigOh{k2^n}$ entries rather than $4^n$), to the \texttt{PauliDecompose} routine, and to our algorithm via the query access model, and we calculate the Pauli decomposition using each approach. The \texttt{PauliDecompose} algorithm requires access to all of the nonzero elements of the matrix to calculate every coefficient, but our algorithm only requires the elements $M_{v,u}$ that we need to query. We compare the wall-clock runtime of both our algorithm and \texttt{PauliDecompose}. As can be seen in~\cref{fig:runtime}, the PennyLane algorithm is faster for small $n$, due to the overhead from our randomized algorithm. However, the PennyLane benchmark algorithm grows rapidly with $n$ and at around $n=10$, there is a crossover, after which our algorithm performs better. Afterwards, we take $n = 13, \ldots, 30$ and only run our algorithm since building the matrix and running \texttt{PauliDecompose} becomes prohibitively costly. Instead, we extrapolate the runtime based on the data for $n = 2, \ldots, 12$. We also numerically confirm the scaling of the query complexity of our algorithm in~\cref{fig:queries}. We fit both the runtime and the query count curves and find both scale polynomially with $n$, which is what we expected. In all runs, our algorithm found the full decomposition without fail, that is, it did not output \texttt{Fail} at any point and did not produce an incorrect decomposition. While perhaps surprising, we attribute the lack of failure to find the correct decomposition with probability $\delta$ to our bounds likely being loose.

\begin{figure*}[tb]
    \centering
    \includegraphics[scale=0.6]{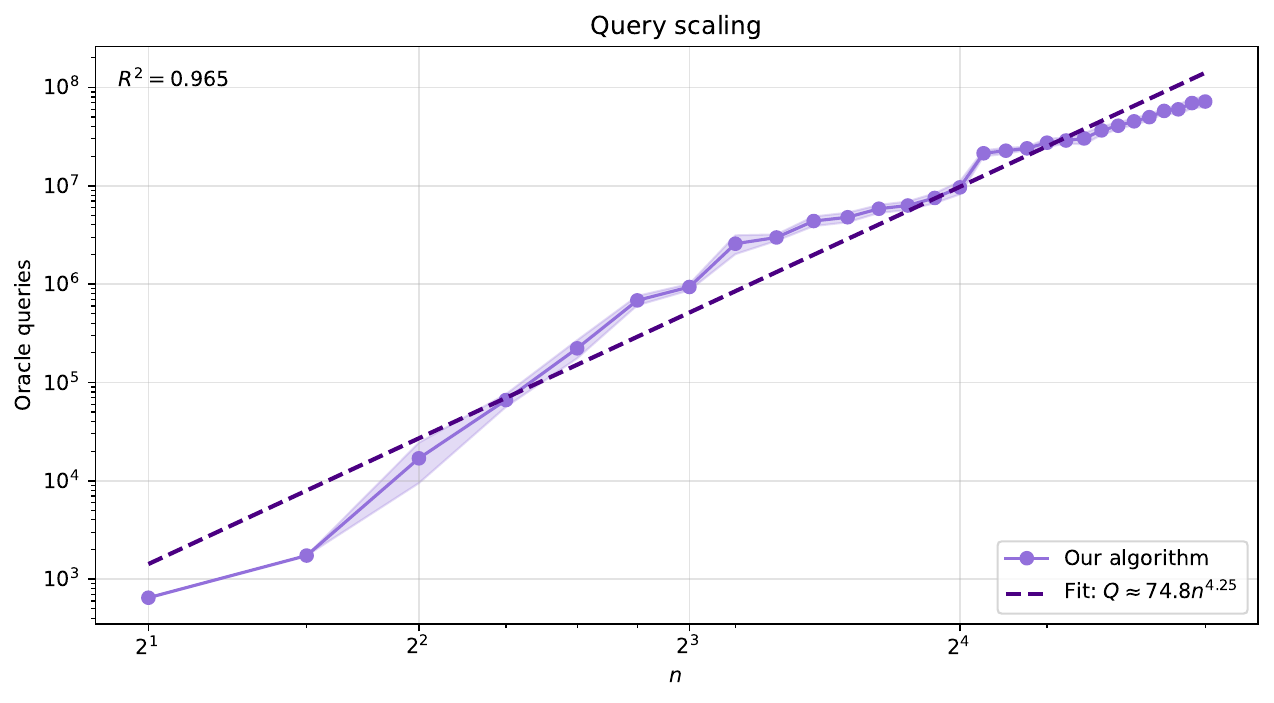}
    \caption{Query count of our algorithm for $n = 2, \ldots, 30$.}
    \label{fig:queries}
\end{figure*}

We acknowledge that our algorithm is tailored for the problem in~\cref{prob:main-problem}; if the input matrix fails to meet the polynomial support problem, then our algorithm is comparable or worse than several of the currently proposed Pauli decomposition algorithms due to the overhead from the randomization required by our approach.

\section{Discussion}\label{sec:discussion}
In this paper, we studied a promise version of the Pauli decomposition problem in which we are given classical sparse query access to a $2^n \times 2^n$ matrix that is assumed to have support on only $k = \poly(n)$ Pauli strings. To find the decomposition, we introduced a randomized classical algorithm that successfully recovers the Pauli decomposition with probability $1-\delta$ for any desired $\delta \in (0,1)$, with query and runtime complexity scaling only polynomially with respect to $n$, $k$, and $\log(1/\delta)$. Thus, although the Pauli decomposition in full generality is an exponentially hard problem, we show that it becomes efficiently solvable in a structured sparse-support regime, which is relevant to some near-term quantum algorithms working with classical input data~\cite{bharti2021iterative,huang2021near-term,bharti2022sdp}.

Conceptually, our algorithm takes advantage of the structure of Pauli strings in binary symplectic form. The Pauli-$X$ part determines where the nonzero matrix elements are, and the Pauli-$Z$ part determines the sign pattern, which leads us to considering two types of $x$ bit strings: unique and degenerate. The unique $x$ bit strings can be decoded directly from a small number of queries, while the degenerate $x$ bit strings are harder to deal with, and the process reduces to a sparse Walsh-Hadamard problem on the Boolean cube.

Our results have a specific scope, namely, those matrices that are assumed to have polynomial support on the Pauli basis and can be accessed via the query access model in~\cref{def:query-access-model}. Our algorithm does not provide a generic efficient Pauli decomposition for arbitrary matrices, as this problem is intractable in full generality. Furthermore, we note that our analysis concerns the query and runtime complexity of the algorithm, where the memory required may still be exponential in $n$.

We highlight the dependence on the promise of $k$-sparsity in the Pauli basis. Given an explicit upper bound $k_{\max}$ satisfying $k_{\max} \geq k$, our algorithm recovers the Pauli decomposition with high probability. If, however, the bound $k_{\max}$ is smaller than the true Pauli support size $k$, then the algorithm may return \texttt{Fail}, but it is not guaranteed to do so. We could supplement the algorithm with a global randomized verification at the end, where, given a candidate decomposition $M_{\mathrm{candidate}}$, we can sample random rows or columns of $M$ and compare the output with what we would expect given $M_{\mathrm{candidate}}$. This procedure can detect a "residual" $R = M - M_{\mathrm{candidate}}$ that is nonzero on an inverse-polynomial fraction of rows or columns, but this type of verification does not protect against arbitrary inputs. For example, consider $M = \ketbra{a}$ for some $a \in \{0,1\}^n$, which has Pauli decomposition
\begin{equation}
    \ketbra{a} = \frac{1}{2^n}\defprod{j=1}{n}{\parens*{I + (-1)^{a_j}Z_j}}.
\end{equation}
This has an exponential number of Pauli terms but the matrix is a matrix of all zeros except for one element, and so making a polynomial number of random checks would very likely miss this single nonzero entry.

There are a few settings where our algorithm can be particularly powerful. The first is when the matrix is specified implicitly by an efficient classical procedure that can answer the required sparse row or column queries without needing to construct the full $2^n \times 2^n$ matrix in memory. For example, suppose we are given a classical function that generates the nonzero entries in any requested row or column of a matrix that is promised to have polynomial support in the Pauli basis. Then, our algorithm could use this procedure as its query oracle to recover the Pauli decomposition using only polynomially many calls to the function, which allows us to avoid materializing the full matrix. This is especially relevant when the matrix is naturally produced by a succinct classical rule, simulation routine, or data-generating process rather than given as an explicit array.

A second useful setting is one in which the matrix is stored in a general query-accessible data structure because there are many possible downstream tasks (perhaps exponentially many) that may be performed on the same data, and the Pauli decomposition is only one such task. For example, the matrix may be updated until shortly before the desired task is specified, or the same stored data may be used for other linear-algebraic or statistical tasks, or other quantum-algorithmic primitives. As such, we would be required to keep the matrix stored in memory.

Finally, our query model may be useful in delegated or adversarial settings. For example, the data may be held by a server that we do not trust to perform the Pauli decomposition itself, while still allowing individual row and column queries to be returned and publicly audited in some way. In this setting, the client can run the decomposition algorithm on their end using only the output of queries to the server holding the data, instead of relying on the server to return the full decomposition, which the client may not trust. This separates the task of storing or serving the data from the task of certifying the Pauli decomposition. 

There are several natural directions for future work. For example, it would be nice to tighten our upper bounds, as we believe these to be somewhat loose due to the various certification steps. We also only focused on upper bounds in this work; it would be good to obtain tight lower bounds beyond a trivial $\bigOmega{kn}$ lower bound to write down the decomposition ($k$ Pauli terms labeled by $(x \mid z) \in \{0,1\}^{2n}$ and their coefficients). Similarly, tightening our (likely loose) upper bounds and tightening the trivial lower bound to saturation would be a fruitful result. Another useful direction would be to consider \emph{approximately} sparse Pauli-support matrices, where most of the support is on a polynomial number of Pauli basis elements, with a small contribution from the remaining basis elements. Progress on these questions would further clarify when classical preprocessing can be made efficient, yielding truly powerful end-to-end near-term quantum algorithms.

\begin{acknowledgments}
    D.J.S.~thanks Adam Ehrenberg for help in the initial development of this project. D.J.S.~further thanks Alexander Dalzell, Weiyuan Gong, and Alexander Schuckert for helpful discussions. D.J.S.~acknowledges funding support from a graduate research fellowship from the Joint Quantum Institute at the University of Maryland, College Park. D.J.S. and A.V.G.~acknowledge support from the U.S.~Department of Energy, Office of Science, Accelerated Research in Quantum Computing, Fundamental Algorithmic Research toward Quantum Utility (FAR-Qu). D.J.S. and A.V.G.~were also supported in part by the NSF STAQ program, DoE ASCR Quantum Testbed Pathfinder program (award No.~DE-SC0024220), NSF QLCI (award No.~OMA-2120757), AFOSR MURI, ONR MURI, DARPA SAVaNT ADVENT, ARL (W911NF-24-2-0107), and NQVL:QSTD:Pilot:FTL. D.J.S. and A.V.G.~also acknowledge support from the U.S.~Department of Energy, Office of Science, National Quantum Information Science Research Centers, Quantum Systems Accelerator (award No.~DE-SCL0000121). K.B.~is supported by a Hartree Fellowship from the Joint Center for Quantum Information and Computer Science (QuICS) at the University of Maryland, College Park.
\end{acknowledgments}

\bibliography{references}
\clearpage
\onecolumngrid
\appendix

\section{Correctness and error analysis}\label{app-sec:correctness-error-analysis}
In this Appendix, we show that our algorithm correctly recovers the Pauli decomposition with high probability. First, recall that for a fixed $x \in \{0,1\}^n$, we define
\begin{equation}
    b_x(u) \coloneqq M_{x \oplus u, u}
\end{equation}
for a query on column $u$ for the fixed $x$. Recall also the notation
\begin{equation}
    \beta_x(z) \coloneqq i^{\abs{x \wedge z}}\alpha_{x,z}
\end{equation}
for coefficients $\alpha_{x,z}$. By~\cref{lemma:pauli-string-matrix-element}, we can write
\begin{equation}
    b_x(u) = \defsum{z \in \{0,1\}^n}{}{\beta_x(z)(-1)^{z \cdot u}},\label{eqn:folding-slice}
\end{equation}
where $\beta_x$ is $p_x$-sparse; in other words, $p_x = \abs{\mathrm{supp}(\beta_x)} \leq k$, where $k$ is the number of Pauli strings the matrix $M$ has support on, and we assume $k$ is known. The first step of the algorithm is to discover the $x$ bit strings present in the Pauli decomposition of $M$; that is, we aim to find the set
\begin{equation}
    \mathcal{X} \coloneqq \{x \in \{0,1\}^n : p_x \geq 1\}.
\end{equation}
We recall the three different error budgets that we make use of throughout our analysis. First, $\eta_{\mathcal{X}}$ denotes the failure probability allowed for the discovery of the active $x$ bit strings, that is, the set $\mathcal{X}$. Second, for each $x \in \mathcal{X}$, we denote by $\eta_x$ the failure probability allowed in the unique-$x$ certification step of~\cref{alg:unique-x-bit-strings}. Finally, for each degenerate $x \in \mathcal{X}$, $\delta_x$ denotes the failure probability allowed in the degenerate recovery step of~\cref{alg:determine-degenerate-x-pauli-string}. We choose these budgets in such a way that
\begin{equation}
    \eta_{\mathcal{X}} + \defsum{x \in \mathcal{X}}{}{\eta_x} + \defsum{x \in \mathcal{X}_{\mathrm{deg}}}{}{\delta_x} \leq \delta,
\end{equation}
where $\delta$ is the failure probability of the whole algorithm that we take as a user-selected input. Since $\abs{\mathcal{X}} \leq k$, a simple choice is
\begin{equation}
    \eta_{\mathcal{X}} = \frac{\delta}{3}, \quad
    \eta_x = \frac{\delta}{3k}, \quad
    \delta_x = \frac{\delta}{3k}.\label{eqn:failure-probabilities}
\end{equation}
With this choice, the total failure probability is at most $\delta$ by a union bound.

As we discussed in~\cref{subsec:unique-x-bit-strings}, while a single query to the first row $v = 0^n$ suffices to find $\mathcal{X}$ for many cases, there are some cases where the structure of the coefficients and $z$ bit strings cancel out the contribution in the first row for a fixed $x \in \mathcal{X}$, as we saw for the simple single-qubit example in the main text. As such, we hedge against this by querying rows uniformly at random:

\begin{lemma}[Discovery of $x$ bit strings]\label{lemma:x-bit-string-discovery}
    Let $\mathcal{X} \subseteq \{0,1\}^n$ be the correct set of $x$ bit strings present in the Pauli decomposition of a matrix $M \in \CC^{2^n \times 2^n}$, where $\abs{\mathcal{X}} \leq k$. Furthermore, let $\eta_{\mathcal{X}} \in (0,1)$ and sample
    \begin{equation}
        R_{\mathcal{X}} \coloneqq \ceil*{k\log\lp \frac{k}{\eta_{\mathcal{X}}} \rp}
    \end{equation}
    rows $v^{(1)}, \ldots, v^{(R_{\mathcal{X}})} \in \{0,1\}^n$ of $M$ uniformly at random. For each queried row $v^{(r)}$, add the bit string $x = u \oplus v^{(r)}$ to the candidate set $\hat{\mathcal{X}} \subseteq \{0,1\}^n$ for every nonzero entry $(u, M_{v^{(r)},u})$ returned by the query. Then,
    \begin{equation}
        \Pr[\hat{\mathcal{X}} = \mathcal{X}] \geq 1 - \eta_{\mathcal{X}}.
    \end{equation}
\end{lemma}

\begin{proof}
    We note that this query-based procedure for discovering the $x$ bit strings never adds a false $x$ bit string, that is, one that isn't in the true decomposition. By~\cref{lemma:row-column-pauli-strings}, a Pauli string described by a bit string $x$ contributes to row $v$ only in the column $u = x \oplus v$, so if $M_{v,u} \neq 0$, then there must be at least one Pauli string described by $x$ present in the decomposition. Thus, we do not have to worry about "false positives."

    We now must simply bound the probability of missing a true $x \in \mathcal{X}$. To do so, we fix an $x \in \mathcal{X}$. Then, for a row $v$ chosen uniformly at random, the corresponding column index is $u = x \oplus v$, which is also uniformly random over $\{0,1\}^n$. The relevant matrix element is $M_{v,x \oplus v} = b_x(x \oplus v)$, which has Walsh support size $p_x \leq k$. By~\cref{fact:boolean-uncertainty-principle}, we have
    \begin{equation}
        \Pr_{v}[M_{v,x \oplus v} \neq 0] = \Pr_{u}[b_x(u) \neq 0] \geq \frac{1}{p_x} \geq \frac{1}{k}.
    \end{equation}
    Thus, the total probability that this given $x$ bit string is missed in all $R_{\mathcal{X}}$ queries that we make is at most
    \begin{equation}
        \lp 1 - \frac{1}{k} \rp^{R_{\mathcal{X}}} \leq e^{-R_{\mathcal{X}}/k} \leq \frac{\eta_{\mathcal{X}}}{k}.
    \end{equation}
    Taking a union bound over at most $k$ active $x$ bit strings gives us
    \begin{equation}
        \Pr[\hat{\mathcal{X}} \neq \mathcal{X}] \leq \eta_{\mathcal{X}},
    \end{equation}
    or, equivalently,
    \begin{equation}
        \Pr[\hat{\mathcal{X}} = \mathcal{X}] \geq 1 - \eta_{\mathcal{X}},
    \end{equation}
    as claimed in the lemma.
\end{proof}

We now prove the correctness of~\cref{alg:unique-x-bit-strings}, which determines whether a fixed $x$ bit string in the discovered set $\mathcal{X}$ is unique and, if so, how to decode it (i.e., get its $z$ bit string and associated coefficient).

\begin{theorem}[Correctness of~\cref{alg:unique-x-bit-strings}]\label{thm:unique-algo-correctness}
    Fix an $x \in \mathcal{X}$ and run~\cref{alg:unique-x-bit-strings} with failure budget $\eta_x \in (0,1)$ and
    \begin{equation}
        L_{\mathrm{cert}} \geq (k+1)\log\lp \frac{1}{\eta_x} \rp
    \end{equation}
    random certification queries. Then,
    \begin{enumerate}
        \item If $p_x = 1$ (i.e., $x$ is unique), the algorithm returns $(\mathtt{Unique}, \hat{z}, \beta)$ with $\hat{z} = z^\star$ and $\beta = \beta_x(z^\star)$ deterministically, where $z^\star$ is the unique support element of $\beta_x$.
        \item If $p_x \geq 2$, the algorithm returns $\mathtt{Degenerate}$ with probability at least $1 - \eta_x$.
    \end{enumerate}
\end{theorem}

\begin{proof}
    First, assume $p_x = 1$. Then, there exists a unique $z^\star \in \{0,1\}^n$ such that $\beta_x(z^\star) \neq 0$, or, in other words,
    \begin{equation}
        b_x(u) = \beta_x(z^\star)(-1)^{z^\star \cdot u}
    \end{equation}
    for all $u \in \{0,1\}^n$. If we take $u = 0^n$, then we have
    \begin{equation}
        b_x(0^n) = \beta_x(z^\star) \neq 0.
    \end{equation}
    Then, for each standard basis vector $e_j$, we have
    \begin{equation}
        \frac{b_x(e_j)}{b_x(0^n)} = \frac{\beta_x(z^\star)(-1)^{z^\star_j}}{\beta_x(z^\star)} = (-1)^{z^\star_j} \in \{\pm 1\}.
    \end{equation}
    Thus, the algorithm reconstructs each bit $z^\star_j$ of $z^\star$ exactly. Finally, every certification query passes because the equality
    \begin{equation}
        b_x(u^{(\ell)}) = \beta_x(z^\star)(-1)^{z^\star \cdot u^{(\ell)}}
    \end{equation}
    holds for all $u^{(\ell)}$ for $\ell \in [L_{\mathrm{cert}}]$.

    Now, assume that $p_x \geq 2$ (i.e., $x$ is actually degenerate). If the algorithm rejects at any of the ratio checks $b_x(e_j)/b_x(0^n) \neq \pm 1$, then the algorithm returns \texttt{Degenerate} and there is nothing else to prove. Thus, we suppose that the basis-query stage produces a candidate $(\hat{z}, \beta)$ with all ratios equal to $\pm 1$. We then define
    \begin{equation}
        h_x(u) \coloneqq b_x(u) - \beta(-1)^{\hat{z} \cdot u}.
    \end{equation}
    If $h_x \equiv 0$, then $b_x$ would be exactly a single Walsh character, which suggests that $\beta_x$ has support on just one element. This would force $p_x = 1$, contradicting the assumption that $p_x \geq 2$, and so $h_x$ is a nonzero function. The Walsh-Hadamard spectrum of $h_x$ is
    \begin{equation}
        \hat{h}_x(z) = \beta_x(z) - \beta\vb{1}_{z=\hat{z}},
    \end{equation}
    where $\vb{1}_{z=\hat{z}}$ is 1 for $z = \hat{z}$ and 0 otherwise. Thus,
    \begin{equation}
        \abs{\mathrm{supp}(\hat{h}_x)} \leq p_x + 1 \leq k + 1.
    \end{equation}
    According to~\cref{fact:boolean-uncertainty-principle}, a uniformly random input $u \in \{0,1\}^n$ satisfies
    \begin{equation}
        \Pr_u[h_x(u) \neq 0] \geq \frac{1}{\abs{\mathrm{supp}(\hat{h}_x)}} \geq \frac{1}{k+1}.
    \end{equation}
    Equivalently, each certification query detects that the candidate singleton is not truly a singleton with probability at least $1/(k+1)$. As each certification query is sampled independently, the probability that all $L_{\mathrm{cert}}$ checks fail to detect the error is at most
    \begin{equation}
        \lp 1 - \frac{1}{k+1} \rp^{L_{\mathrm{cert}}} \leq e^{-L_{\mathrm{cert}}/(k+1)} \leq \eta_x.
    \end{equation}
    where the second inequality is obtained by choosing
    \begin{equation}
        L_{\mathrm{cert}} \geq (k+1)\log\lp \frac{1}{\eta_x} \rp.
    \end{equation}
    Thus, when $p_x \geq 2$,~\cref{alg:unique-x-bit-strings} returns \texttt{Degenerate} with probability at least $1 - \eta_x$.
\end{proof}

\cref{alg:determine-degenerate-x-pauli-string} maintains a recovered set $\mathcal{S}$ of pairs $(z,\beta_x(z))$, where recall, given $\mathcal{S}$, we define the residual oracle as
\begin{equation}
    \Tilde{b}_x(u) \coloneqq b_x(u) - \defsum{(z,\beta)\in\mathcal{S}}{}{\beta(-1)^{z \cdot u}}.
\end{equation}
When all elements of $\mathcal{S}$ are correct, $\Tilde{b}_x$ has Walsh-Hadamard spectrum equal to the remaining unrecovered part of $\beta_x$. Denote by $p_\tau$ the number of remaining true support elements to be found at the beginning of round $\tau$, conditioned on the event that no false support element has been accepted. Then,~\cref{lemma:folding-identity} in~\cref{subsec:walsh-hadamard-analysis} tells us that the queried values
\begin{equation}
    y(t) = \Tilde{b}_x(R^\intercal t)
\end{equation}
have Walsh-Hadamard spectrum
\begin{equation}
    g(s) = \defsum{z:Rz=s}{}{\Tilde{\beta}_x(z)}.
\end{equation}
Furthermore,~\cref{lemma:folding-identity} implies that a set of shifted queries satisfies
\begin{equation}
    g^{(w)}(s) = \defsum{z:Rz=s}{}{\Tilde{\beta}_x(z)(-1)^{z \cdot w}}.
\end{equation}
For the remainder of this section, we analyze the degenerate $x$ bit string algorithm in~\cref{alg:determine-degenerate-x-pauli-string}. To do so, we fix an $x \in \mathcal{X}$ such that $p_x \geq 2$. Furthermore, let $p_x^{\max}$ be an upper bound on $p_x$ such that $p_x \leq p_x^{\max}$; the trivial choice $p_x^{\max} = k$ suffices. \cref{alg:determine-degenerate-x-pauli-string} uses a per-slice failure budget $\delta_x$ for a given $x$ bit string, which we split up across three possible failure events:
\begin{enumerate}
    \item Some true residual support element is never isolated in any of the folding rounds;
    \item Some collision bin is incorrectly accepted as a singleton;
    \item After the peeling rounds, the residual is nonzero but the final randomized residual check fails to detect it.
\end{enumerate}
We allocate failure probability at most $\delta_x/3$ to each of these events for simplicity. 

If a bin $s$ contains exactly one support element $z^\star$, then $g(s) = \beta_x(z^\star)$ and $g^{(w)}(s) = \beta_x(z^\star)(-1)^{z^\star \cdot w}$ for all shifts $w$. Choosing $w = e_j$ for $j \in [n]$ and taking the ratio
\begin{equation}
    \frac{g^{(e_j)}(s)}{g(s)} = (-1)^{z^\star \cdot e_j} = (-1)^{z^\star_j}
\end{equation}
then allows us to decode $z^\star$ bit-by-bit, as we saw for the unique $x$ bit strings above. We determine the probability of isolating such a $z^\star$ in a given round according to the following lemma:

\begin{lemma}[Finite-$p_\tau$ singleton bound]\label{lemma:finite-singleton-bound}
    Fix one unrecovered support element $z^\star$ at the beginning of a round. If the fold uses $B = 2^m$ bins, then
    \begin{equation}
        \Pr[z^\star \text{ collides with at least one other unrecovered element}] \leq \frac{p_\tau-1}{B},
    \end{equation}
    or, in other words,
    \begin{equation}
        \Pr[z^\star \text{ is a singleton}] \geq 1 - \frac{p_\tau-1}{B}.
    \end{equation}
    If $B \geq cp_\tau$ for some constant $c$, then
    \begin{equation}
        \Pr[z^\star \text{ is a singleton}] \geq 1 - \frac{p_\tau-1}{B} \geq 1 - \frac{1}{c} + \frac{1}{cp_\tau}.
    \end{equation}
\end{lemma}

\begin{proof}
    For any distinct $z \neq z^\star$, to have $Rz = Rz^\star$ is equivalent to $R(z \oplus z^\star) = 0$. Since $z \oplus z^\star \neq 0^n$ by the fact that $z \neq z^\star$ and each row of $R$ is uniformly sampled from $\{0,1\}^n$, we have
    \begin{equation}
        \Pr[Rz = Rz^\star] = 2^{-m} = \frac{1}{B}.
    \end{equation}
    Taking a union bound over the other $p_\tau - 1$ unrecovered elements then gives us
    \begin{equation}
        \Pr[z^\star\text{ is a singleton}] \geq 1 - \frac{p_\tau-1}{B}.
    \end{equation}
    Then, if we take $B \geq cp_\tau$, we can conclude that
    \begin{align}
        \Pr[z^\star\text{ is a singleton}] &\geq 1 - \frac{p_\tau-1}{cp_\tau}\\
        &= 1 - \frac{1}{c} + \frac{1}{cp_\tau},
    \end{align}
    as claimed in the lemma.
\end{proof}

With the probability of a singleton being isolated in a single round established, we now calculate the expected number of true (i.e., not mistaken singletons, which we analyze later) singleton bins decoded in round $\tau$, which we denote by the random variable $X_\tau$. Thus, we have
\begin{equation}
    p_{\tau+1} = p_\tau - X_{\tau+1}
\end{equation}
and the following lemma:

\begin{lemma}[Progress over repeated folds]\label{lemma:expected-contraction}
    Fix a degenerate $x$ bit string with degeneracy $p_x \geq 2$ and suppose that the folding procedure uses $B = 2^m$ bins with $B \geq cp_x$ for some constant $c > 1$. Condition on the event that no collision bin is falsely accepted as a singleton during the folding rounds. Let $p_0 \coloneqq p_x$ and denote by $p_\tau$ the number of unrecovered true support elements after $\tau$ folding rounds. Then, for every $\tau \geq 0$,
    \begin{equation}
        \mathbb{E}[p_{\tau+1} \mid p_\tau] \leq \frac{p_\tau}{c}
    \end{equation}
    Thus, after $T$ rounds,
    \begin{equation}
        \mathbb{E}[p_T] \leq \frac{p_x}{c^T} \leq \frac{p^{\max}_x}{c^T.}
    \end{equation}
    Therefore,
    \begin{equation}
        \Pr[p_T \geq 1] \leq \frac{p_x^{\max}}{c^T}.
    \end{equation}
    Choosing $T \geq \ceil*{\log_c{\parens*{\frac{3p_x^{\max}}{\delta_x}}}}$ thus ensures that all support elements are isolated in at least one round with probability at least $1 - \delta_x/3$.
\end{lemma}

\begin{proof}
    Fix a round $\tau + 1$ and suppose that $p_\tau$ true support elements remain unrecovered at the beginning of this round. If $p_{\tau} = 0$, then $p_{\tau+1} = 0$ deterministically, and the desired bound is immediate. Thus, for the rest of the proof, we assume that $p_\tau \geq 1$. Index the remaining support elements to be found by $z^{(1)}, \ldots, z^{(p_\tau)}$. For each $i \in [p_\tau]$, denote by $I_i$ the indicator random variable that $z^{(i)}$ is isolated as a singleton in round $\tau + 1$. Then, $X_{\tau+1} := \defsum{i=1}{p_\tau}{I_i}$ is the number of true support elements isolated as singletons in this round. By~\cref{lemma:finite-singleton-bound}, each remaining support element is isolated with probability at least $1 - \frac{p_\tau-1}{B}$. Using the linearity of expectation, we have
    \begin{equation}
        \mathbb{E}[X_{\tau+1} \mid p_\tau] = \defsum{i=1}{p_\tau}{\Pr[I_i=1 \mid p_\tau]} \geq p_\tau\lp 1 - \frac{p_\tau-1}{B} \rp.
    \end{equation}
    The number of unrecovered elements after the round is then $p_{\tau+1} = p_\tau - X_{\tau+1}$. Thus, we have
    \begin{equation}
        \mathbb{E}[p_{\tau+1} \mid p_\tau] = p_\tau - \mathbb{E}[X_\tau \mid p_\tau] \leq p_\tau \cdot \frac{p_\tau-1}{B}.
    \end{equation}
    Since $p_\tau \leq p_x$ and $B \geq cp_x$, it is also true that $B \geq cp_\tau$, so we have
    \begin{equation}
        p_\tau \cdot \frac{p_\tau-1}{B} \leq p_\tau \cdot \frac{p_\tau-1}{cp_\tau} = \frac{p_\tau-1}{c} \leq \frac{p_\tau}{c}.
    \end{equation}
    Therefore,
    \begin{equation}
        \mathbb{E}[p_{\tau+1} \mid p_\tau] \leq \frac{p_\tau}{c}.
    \end{equation}
    Then, taking the expectation of the above quantity, we have
    \begin{equation}
        \mathbb{E}[p_{\tau+1}] = \mathbb{E}[\mathbb{E}[p_{\tau+1} \mid p_\tau]] \leq \mathbb{E}\lb \frac{p_\tau}{c} \rb = \frac{1}{c}\mathbb{E}[p_\tau].
    \end{equation}
    By induction and $p_0 = p_x$, by round $\tau = T$, we have
    \begin{equation}
        \mathbb{E}[p_T] \leq \frac{p_x}{c^T} \leq \frac{p^{\max}_x}{c^T}.
    \end{equation}
    Finally, using Markov's inequality, we have
    \begin{equation}
        \Pr[p_T \geq 1] \leq \mathbb{E}[p_T] \leq \frac{p_x^{\max}}{c^T}.
    \end{equation}
    Thus, choosing $T \geq \ceil*{\log_c\lp \frac{3p_x^{\max}}{\delta_x} \rp}$ ensures that
    \begin{equation}
        \frac{p_x^{\max}}{c^T} \leq \frac{\delta_x}{3},
    \end{equation}
    which proves the claim in the lemma.
\end{proof}

We mentioned above that there is a chance that a collision bin \emph{appears} to be a singleton and passes the random certification shifts. We now bound that probability. First, fix a bin $s$ in some round. Suppose also that the true set of coefficients hashed into bin $s$ is $\{(z^{(1)}, \beta_1), \ldots, (z^{(r)}, \beta_r)\}$ for $r \geq 2$. Furthermore, assume that the basis-shift ratios produce a candidate singleton $(\hat{z}, \hat{\beta})$. Define the quantity
\begin{equation}
    h(w) \coloneqq g^{(w)}(s) - \hat{\beta} \cdot (-1)^{\hat{z} \cdot w}.\label{eqn:discrepancy-function}
\end{equation}
If the candidate is wrong, then $h$ is nonzero with Walsh support size is at most $r+1$. Then, we have the following:

\begin{lemma}[Random shift certification]\label{lemma:random-shift-certification}
    If a candidate singleton for a bin is incorrect and the true bin contains $r \geq 2$ elements, then a uniformly random shift $w$ rejects the candidate with probability at least $1/(r+1) \geq 1/(p^{\max}_x+1)$. Therefore, $L$ independent random shifts reject the candidate with probability at least
    \begin{equation}
        1 - \lp 1 - \frac{1}{p^{\max}_x+1} \rp^L.
    \end{equation}
\end{lemma}

\begin{proof}
    The shifted folded value of bin $s$ is
    \begin{equation}
        g^{(w)}(s) = \defsum{i=1}{r}{\beta_i(-1)^{z^{(i)} \cdot w}}.
    \end{equation}
    The candidate singleton instead would suggest that the shifted folded value should be $\hat{\beta}(-1)^{\hat{z} \cdot w}$. With the discrepancy function defined in~\cref{eqn:discrepancy-function} and the expression for $g^{(w)}(s)$ above, we have
    \begin{equation}
        h(w) = \defsum{i=1}{r}{\beta_i(-1)^{z^{(i)} \cdot w} - \hat{\beta}(-1)^{\hat{z} \cdot w}}.
    \end{equation}
    Thus, $h$ is a Walsh-sparse function of the shift $w$ and its Walsh support is contained in $\{z^{(1)}, \ldots, z^{(r)}\} \cup \{\hat{z}\}$. That is, $\abs{\mathrm{supp}(\hat{h})} \leq r+1$. Furthermore, $h$ is a nonzero function, because if $h \equiv 0$, then
    \begin{equation}
        \defsum{i=1}{r}{\beta_i(-1)^{z^{(i)} \cdot w}} = \hat{\beta}(-1)^{\hat{z} \cdot w}
    \end{equation}
    for all $w \in \{0,1\}^n$. By the uniqueness of the Walsh-Hadamard expansion, the true bin contribution would be only a single Walsh character with coefficient $\hat{\beta}$, but this contradicts the assumption that the bin contains $r \geq 2$ nonzero residual support elements. Thus, $h \not\equiv 0$ and has Walsh support size at most $r+1$, so by~\cref{fact:boolean-uncertainty-principle}, we have
    \begin{equation}
        \Pr_{w\sim\{0,1\}^n}[h(w) \neq 0] \geq \frac{1}{r+1}.
    \end{equation}
    As $r \leq p_x \leq p^{\max}_x$, it is true that
    \begin{equation}
        \Pr_w[h(w) \neq 0] \geq \frac{1}{p^{\max}_x + 1}.
    \end{equation}
    Note that the random-shift certification test accepts the candidate on shift $w$ when
    \begin{equation}
        \frac{g^{(w)}(s)}{g(s)} = (-1)^{\hat{z} \cdot w}.
    \end{equation}
    Since $\hat{\beta} = g(s) \neq 0$, this is equivalent to
    \begin{equation}
        g^{(w)}(s) - \hat{\beta}(-1)^{\hat{z} \cdot w} = 0,
    \end{equation}
    which is exactly the same as $h(w) = 0$. Thus, when $h(w) \neq 0$, the random-shift test rejects the candidate, with a single test rejecting with probability at least $1/(p_x^{\max}+1)$. Over $L$ independent shift tests, the probability that none of them reject the candidate is at most
    \begin{equation}
        \parens*{1 - \frac{1}{p^{\max}_x+1}}^L,
    \end{equation}
    and so the probability that at least one of the $L$ tests rejects the candidate is at least
    \begin{equation}
        1 - \parens*{1 - \frac{1}{p_x^{\max}+1}}^L,
    \end{equation}
    as claimed in the lemma.
\end{proof}

If we choose
\begin{equation}
    L \geq (p^{\max}_x+1)\log\lp \frac{3BT}{\delta_x} \rp,
\end{equation}
then the probability that any one of the at most $BT$ candidate bins in a fixed-$x$ run passes all certification despite being incorrect is at most $\delta_x/3$.

Finally, after obtaining a candidate set $\mathcal{S}$ for the fixed $x$ bit string, we perform a final check to ensure that we did not miss any support elements. Formally,

\begin{lemma}[Final residual check]\label{lemma:final-residual-check}
    Suppose that after the peeling rounds the residual $\Tilde{b}_x$ is nonzero and that no false support elements have been accepted. Then $\Tilde{b}_x$ has a Walsh support of size at most $p_x^{\max}$, and if we sample
    \begin{equation}
        q \geq p_x^{\max}\log\lp \frac{3}{\delta_x} \rp
    \end{equation}
    inputs $u^{(1)}, \ldots, u^{(q)}$ uniformly at random, then
    \begin{equation}
        \Pr[\Tilde{b}_x(u^{(\ell)}) = 0 \text{ for all } \ell \in [q]] \leq \frac{\delta_x}{3}.
    \end{equation}
\end{lemma}

\begin{proof}
    Note that $\Tilde{b}_x$ has a nonzero residual Walsh-Hadamard spectrum with support size at most $p_x^{\max}$. By~\cref{fact:boolean-uncertainty-principle}, if we make a query on a uniformly sampled location, then we can detect the residual with probability at least $1/p_x^{\max}$. Thus, the probability that all $q$ queries miss the residual is at most
    \begin{equation}
        \lp 1 - \frac{1}{p_x^{\max}} \rp^q \leq e^{-q/p_x^{\max}} \leq \frac{\delta_x}{3},
    \end{equation}
    which proves the claim.
\end{proof}

With all of the above established, we make our final statement on the success probability for decoding a single fixed degenerate $x$ bit string. By combining~ all of the above analysis, we have the following theorem:

\begin{theorem}[Fixed-$x$ recovery probability]\label{thm:fixed-x-recovery-probability}
    Fix a degenerate $x$ bit string and let $p_x^{\max} \geq p_x$. Choose
    \begin{align}
        m &= \ceil*{\log(cp_x^{\max})},\quad B = 2^m,\\
        T &= \ceil*{\log_c\lp \frac{3p_x^{\max}}{\delta_x} \rp},\\
        L &= \ceil*{(p_x^{\max}+1)\log\lp \frac{3BT}{\delta_x} \rp},\\
        q &= \ceil*{p_x^{\max}\log\lp \frac{3}{\delta_x} \rp}.
    \end{align}
    Then,~\cref{alg:determine-degenerate-x-pauli-string} returns the correct support and coefficients of $\beta_x$ with probability at least $1 - \delta_x$.
\end{theorem}

\begin{proof}
    Consider the failure events mentioned above, before~\cref{lemma:finite-singleton-bound}. We showed that each of these events occurs with probability at most $\delta_x/3$. If none of these events occur, then every true support element is isolated in some round, every accepted candidate is correct, and the final residual check confirms that no true support element remains undiscovered. Thus, the set $\mathcal{S}$ that we find is exactly the true support and coefficient set of $\beta_x$, and taking a union bound over all three events gives us a total failure probability of at most $\delta_x$, or a success probability at least $1-\delta_x$.
\end{proof}

\section{Query complexity}\label{app-sec:query-complexity}
In this Appendix, we bound the number of queries made to the matrix in order to solve~\cref{prob:main-problem}. We do not attempt to optimize constants or reuse duplicate queries; such optimization can be done at the point of implementation of the algorithm. Furthermore, we note that the powers of $n$, $k$, and $\log(1/\delta)$ \emph{can} be further optimized slightly for both the query complexity and the runtime complexity in~\cref{app-sec:time-complexity}, but we omit such optimization for clarity of exposition. For example, we take the runtime of $n$-bit string arithmetic to be $\bigOh{n}$, where this can be improved using something like bit-packing. However, our primary goal is to show that the number of queries and the runtime scale polynomially in $n$, $k$, and $\log(1/\delta)$, so we accept slightly less optimal bounds in exchange for ease of understanding.

Recall that we denote $\mathcal{X} \coloneqq \{x \in \{0,1\}^n : p_x \geq 1\}$ as the set of distinct $x$ bit strings that appear in the decomposition, where $p_x \coloneqq \abs{\mathrm{supp}(\beta_x)}$ is the number of Pauli strings with Pauli-$X$ part described by $x$. Then, $\abs{\mathcal{X}} \leq k$ and
\begin{equation}
    \defsum{x\in\mathcal{X}}{}{p_x} \leq k.
\end{equation}
For each degenerate $x$, we run~\cref{alg:determine-degenerate-x-pauli-string} with an upper bound $p_x^{\max} \geq p_x$, where a trivial choice is $p_x^{\max} = k$, which is always valid. We make use of the failure budgets from~\cref{app-sec:correctness-error-analysis}, where, recall, we have
\begin{equation}
    \eta_{\mathcal{X}} + \defsum{x\in\mathcal{X}}{}{\eta_x} + \defsum{x\in\mathcal{X}_{\mathrm{deg}}}{}{\delta_x} \leq \delta.    
\end{equation}
A simple choice for these failure budgets satisfying the above constraint and using the trivial upper bound $p_x^{\max}= k$ is given in~\cref{eqn:failure-probabilities}. We state the following theorem on the total query complexity using this trivial bound.

\begin{theorem}[Query complexity]\label{thm:query-complexity}
    Let $p^{\max}_x \geq p_x$ be the upper bound for each degenerate $x$ bit string. Then, using the trivial bound $p^{\max}_x = k$ for all degenerate $x$ bit strings,~\cref{alg:full-algo} solves~\cref{prob:main-problem} with
    \begin{equation}
        \bigOh{nk^2\log{\parens*{\frac{k}{\delta}}} + k^3\log^2{\parens*{\frac{k}{\delta}}}}.
    \end{equation}
    queries to the matrix. Since $k = \poly(n)$, the query complexity is polynomial in $n$, $k$, and $\log(1/\delta)$.
\end{theorem}

\begin{proof}
    We analyze the total number of queries step-by-step, walking through the algorithm as presented in the main text. First, by~\cref{lemma:x-bit-string-discovery}, we make $R_{\mathcal{X}}$ row queries, where
    \begin{equation}
        R_{\mathcal{X}} = \ceil*{k\log{\parens*{\frac{k}{\eta_{\mathcal{X}}}}}}.
    \end{equation}
    Taking $\eta_{\mathcal{X}} = \delta/3$ as in~\cref{eqn:failure-probabilities}, we have
    \begin{equation}
        R_{\mathcal{X}} = \bigOh{k\log{\parens*{\frac{k}{\delta}}}}.
    \end{equation}
    Next, we quantify how many queries are needed to classify each $x$ as unique or degenerate. For each $x \in \mathcal{X}$,~\cref{alg:unique-x-bit-strings} queries $n$ rows corresponding to the values of $u$ that are Hamming weight-1. Furthermore, we make
    \begin{equation}
        L_{\mathrm{cert}} = \bigOh{(k+1)\log\lp \frac{1}{\eta_x} \rp}
    \end{equation}
    random certification queries. Thus, for each $x \in \mathcal{X}$, we make a total of
    \begin{equation}
        \bigOh{n + L_{\mathrm{cert}}} = \bigOh{n + k\log{\parens*{\frac{1}{\eta_x}}}}
    \end{equation}
    queries. As there are $\abs{\mathcal{X}} \leq k$ distinct $x$ bit strings, this step requires a total of
    \begin{equation}
        \bigOh{kn + k^2\log{\parens*{\frac{1}{\eta_x}}}}
    \end{equation}
    queries. Choosing $\eta_x = \delta/(3k)$ according to~\cref{eqn:failure-probabilities}, we have
    \begin{equation}
        \bigOh{kn+k^2\log{\parens*{\frac{k}{\delta}}}}
    \end{equation}
    queries for this step.

    Finally, we determine the number of queries needed for the recovery of the $z$ bit strings and coefficients of degenerate $x$ bit strings. Denote $\mathcal{X}_{\mathrm{deg}}$ as the set of degenerate $x$ bit strings, where $\abs{\mathcal{X}_{\mathrm{deg}}} \leq \abs{\mathcal{X}}$, and fix a degenerate $x \in \mathcal{X}_{\mathrm{deg}}$ with degeneracy $p_x \leq p^{\max}_x$. Recall,  in~\cref{alg:determine-degenerate-x-pauli-string}, we choose $m = \ceil{\log{(cp^{\max}_x)}}$ for some constant $c > 1$ and we perform $T$ rounds of folded inverse Walsh-Hadamard transforms of size $B = 2^m = \bigTheta{p^{\max}_x}$, where
    \begin{equation}
        T = \ceil*{\log_c{\parens*{\frac{3p^{\max}_x}{\delta_x}}}}.
    \end{equation}
    Furthermore, we define quantities
    \begin{equation}
        L = \ceil*{(p^{\max}_x+1)\log{\parens*{\frac{3BT}{\delta_x}}}},\qquad q = \ceil*{p^{\max}_x\log{\parens*{\frac{3}{\delta_x}}}}
    \end{equation}
    for certification. In each round, we query $B$ unshifted points, $nB$ basis-shifted points, and $LB$ random certification shifts. This gives us a total of $(n+1+L)B$ queries for one round. Over the $T$ independent rounds and the $q$ residual checks after the $T$ rounds, the query complexity for a single degenerate $x$ bit string is
    \begin{equation}
        \bigOh{(n+1+L)BT+q}.
    \end{equation}
    Using the values for $L$, $B$, $T$, and $q$ defined above, we have the query complexity scaling as
    \begin{equation}
        \bigOh{np^{\max}_x\log{\parens*{\frac{k}{\delta}}} + (p^{\max}_x)^2\log^2{\parens*{\frac{k}{\delta}}}},
    \end{equation}
    where we have suppressed lower-order terms and log-log terms. Summing over all degenerate $x \in \mathcal{X}_{\mathrm{deg}}$ bit strings, we require
    \begin{equation}
        \bigOh{n\log{\parens*{\frac{k}{\delta}}}\defsum{x\in\mathcal{X}_{\mathrm{deg}}}{}{p^{\max}_x} + \log^2{\parens*{\frac{k}{\delta}}}\defsum{x\in\mathcal{X}_{\mathrm{deg}}}{}{(p^{\max}_x)^2}}
    \end{equation}
    queries. Combining this with the support-discovery queries and the queries used to decode the unique support elements, we have a total of
    \begin{equation}
        \bigOh{k\log{\parens*{\frac{k}{\delta}}} + k\parens*{n+k\log{\parens*{\frac{k}{\delta}}}} + n\log{\parens*{\frac{k}{\delta}}}\defsum{x\in\mathcal{X}_{\mathrm{deg}}}{}{p^{\max}_x} + \log^2{\parens*{\frac{k}{\delta}}}\defsum{x\in\mathcal{X}_{\mathrm{deg}}}{}{(p^{\max}_x)^2}}\label{eqn:messy-query-complexity}
    \end{equation}
    queries. We can clean this up a bit, where we take $p^{\max}_x = k$ and $\abs{\mathcal{X}_{\mathrm{deg}}} \leq k$, so that
    \begin{equation}
        \defsum{x\in\mathcal{X}_{\mathrm{deg}}}{}{p^{\max}_x} \leq k^2\label{eqn:sum-px}
    \end{equation}
    and
    \begin{equation}
        \defsum{x\in\mathcal{X}_{\mathrm{deg}}}{}{(p^{\max}_x)^2} \leq k^3.\label{eqn:sum-px-squared}
    \end{equation}
    Substituting in these bounds to~\cref{eqn:messy-query-complexity} yields a query complexity of
    \begin{equation}
        \bigOh{nk^2\log{\parens*{\frac{k}{\delta}}} + k^3\log^2{\parens*{\frac{k}{\delta}}}}.
    \end{equation}
    Since $k = \poly(n)$, we can see that the query complexity is still polynomial in $n$, $k$, and $\log(1/\delta)$, as claimed.
\end{proof}

As we can see, there is a decent overhead in the number of queries that we need to make for the various levels of randomization, certifications, etc. needed in our algorithm. This overhead is reflected in the numerical experiments that we presented in~\cref{sec:numerics}, where for small values of $n$, the PennyLane benchmark algorithm saw better performance in runtime. However, asymptotically in $n$, since our algorithm is polynomial in $n$, we eventually see a crossover point (in our experiments, this was at $n=8$), after which our algorithm outperforms the benchmark algorithm, as we would expect.

\section{Time complexity}\label{app-sec:time-complexity}
We now prove the runtime of~\cref{alg:full-algo}, which, as can be inferred from the number of queries, is also polynomial in $n$, $k$, and $\log(1/\delta)$. Formally, we have:

\begin{theorem}[Runtime complexity]\label{thm:runtime-complexity}
    Taking $p^{\max}_x = k$,~\cref{alg:full-algo} runs in time
    \begin{equation}
        \bigOh{n^2k^2\log{\parens*{\frac{k}{\delta}}} + nk^3\log^2{\parens*{\frac{k}{\delta}}}}.
    \end{equation}
    As $k = \poly(n)$, the runtime is polynomial in $n$ and $\log(1/\delta)$.
\end{theorem}

\begin{proof}
    We prove the runtime complexity by walking through the algorithm as we did above to prove the query complexity. First, the set of $x$ bit strings $\mathcal{X}$ is discovered using a randomized mechanism, which requires querying
    \begin{equation}
        R_{\mathcal{X}} = \ceil*{k\log{\parens*{\frac{k}{\eta_{\mathcal{X}}}}}}
    \end{equation}
    rows. As we took $\eta_{\mathcal{X}} = \delta/3$ in~\cref{eqn:failure-probabilities}, we have $R_{\mathcal{X}} = \bigOh{k\log(k/\delta)}$. Note that each queried row contains at most $k$ nonzero entries. For each entry in the query output, $(u,M_{v,u})$, we calculate
    \begin{equation}
        x = u \oplus v,
    \end{equation}
    which takes time $\bigOh{n}$ as we need to perform addition modulo 2 for $n$-bit strings. Thus, doing this at most $k$ times, this $x$-support discovery algorithm takes total time
    \begin{equation}
        \bigOh{nkR_{\mathcal{X}}} = \bigOh{nk^2\log{\parens*{\frac{k}{\delta}}}}.
    \end{equation}
    Then,~\cref{alg:unique-x-bit-strings} classifies and decodes unique $x$ bit string, where for each $x \in \mathcal{X}$, we evaluate $b_x(u)$ at $u=0^n$, at $n$ values of $u$ with Hamming weight 1 (i.e., the standard basis vectors), and at
    \begin{equation}
        L_{\mathrm{cert}} = \ceil*{(k+1)\log{\parens*{\frac{1}{\eta_x}}}}
    \end{equation}
    points to check if it is truly unique. Thus, for each $x \in \mathcal{X}$, we have a total of $\bigOh{n+L_{\mathrm{cert}}}$ evaluations, where each evaluation involves calculating the row index $x \oplus u$, looking up the requested column in the returned sparse row, computing the Walsh sign, and performing basic arithmetic, for a total time cost of $\bigOh{n}$ per evaluation. As such, the classification of each $x$ bit string as unique and subsequent decoding of unique ones takes a total time
    \begin{equation}
        \bigOh{\abs{\mathcal{X}}n(n+L_{\mathrm{cert}})}.
    \end{equation}
    We know that $\abs{\mathcal{X}} \leq k$ and $\eta_x = \delta/(3k)$ by~\cref{eqn:failure-probabilities}, so substituting these values in we have $L_{\mathrm{cert}} = \bigOh{k\log{(k/\delta)}}$ and so the total time for this stage is
    \begin{equation}
        \bigOh{kn^2 + nk^2\log{\parens*{\frac{k}{\delta}}}}.
    \end{equation}
    
    Finally, we focus our attention on the degenerate $x$ bit strings. First, fix a degenerate $x$ with degeneracy $p_x \leq p^{\max}_x$. Recall, for~\cref{alg:determine-degenerate-x-pauli-string}, we choose $B = 2^m$ bins with $m = \log{p^{\max}_x}$ such that $B = \bigTheta{p^{\max}_x}$. Then, we choose $T = \bigOh{\log{\parens*{\frac{k}{\delta}}}}$ rounds and $L = \bigOh{p^{\max}_x\log{\parens*{\frac{k}{\delta}}}}$ random certification shifts to eliminate false singletons, where we used $p^{\max}_x \leq k$ and $\delta_x = \delta/(3k)$.
    In each round, we compute one unshifted folded spectrum, $n$ basis-shift folded spectra, and $L$ certification folded spectra, giving us a total of $n+1+L$ folded spectra to deal with. Plugging in for $L$, this is
    \begin{equation}
        \bigOh{n+p^{\max}_x\log{\parens*{\frac{k}{\delta}}}}\label{eqn:number-folded-spectra}
    \end{equation}
    folded spectra. For each such folded spectrum (i.e., call to the $\mathrm{Fold}(x,\mathcal{S},R,w)$ function we defined in~\cref{eqn:fold-function}), we have to evaluate $b_x(R^\intercal t \oplus w)$ on all $B$ values of $t$, perform an $m$-bit inverse Walsh-Hadamard transform, and subtract the terms that we have already recovered. Since each Walsh-Hadamard transform takes time $\bigOh{B\log{B}}$, and since $B = \bigTheta{p^{\max}_x}$, each folded spectrum takes time
    \begin{equation}
        \bigOh{p^{\max}_x(n+\log{p^{\max}_x})}\label{eqn:runtime-cost-per-folded-spectra}
    \end{equation}
    to calculate. Combining~\cref{eqn:number-folded-spectra,eqn:runtime-cost-per-folded-spectra}, calculating all folded spectra in one round costs
    \begin{equation}
        \bigOh{n^2p^{\max}_x + n(p^{\max}_x)^2\log{\parens*{\frac{k}{\delta}}}}
    \end{equation}
    where we have suppressed lower order terms. Across $T$ rounds, the runtime for a fixed degenerate $x$ bit string tallies out to
    \begin{equation}
        \bigOh{n^2p^{\max}_x\log{\parens*{\frac{k}{\delta}}} + n(p^{\max}_x)^2\log^2{\parens*{\frac{k}{\delta}}}}.
    \end{equation}
    The final residual check, recall, uses
    \begin{equation}
        q = \bigOh{p^{\max}_x\log{\parens*{\frac{k}{\delta}}}}
    \end{equation}
    additional evaluations, but this is dominated by the above bound, so we don't include this in the final runtime. If we sum over all degenerate $x \in \mathcal{X}_{\mathrm{deg}}$, then the total runtime to recover these degenerate $x$ bit strings is
    \begin{equation}
        \bigOh{n^2\log{\parens*{\frac{k}{\delta}}}\defsum{x\in\mathcal{X}_{\mathrm{deg}}}{}{p^{\max}_x} + n\log^2{\parens*{\frac{k}{\delta}}}\defsum{x\in\mathcal{X}_{\mathrm{deg}}}{}{(p^{\max}_x)^2}}.
    \end{equation}
    If we combine this with the previous steps to find the $x$ bit strings, determine the uniqueness, and decode the unique $x$ bit strings, we have a total runtime of
    \begin{equation}
        \bigOh{kn^2 + nk^2\log{\parens*{\frac{k}{\delta}}} + n^2\log{\parens*{\frac{k}{\delta}}}\defsum{x\in\mathcal{X}_{\mathrm{deg}}}{}{p^{\max}_x} + n\log^2{\parens*{\frac{k}{\delta}}}\defsum{x\in\mathcal{X}_{\mathrm{deg}}}{}{(p^{\max}_x)^2}}.
    \end{equation}
    Now, we consider the two levels of knowledge we have of the degeneracy of each $x$ bit string. Taking the trivial upper bound $p^{\max}_x = k$ and using~\cref{eqn:sum-px,eqn:sum-px-squared}, the total runtime becomes
    \begin{equation}
        \bigOh{n^2k^2\log{\parens*{\frac{k}{\delta}}} + nk^3\log^2{\parens*{\frac{k}{\delta}}}},
    \end{equation}
    thus proving the claim in the theorem.
\end{proof}

This is polynomial in $n$, $k$, and $\log(1/\delta)$, and since $k = \poly(n)$ and $\delta$ can be chosen such that the runtime remains at most polynomial in $n$, we conclude that the total runtime of the algorithm is at most polynomial in $n$.

\section{Overall guarantee}\label{app-sec:overall-guarantee}
Combining the above analysis of correctness and complexity from~\cref{app-sec:correctness-error-analysis,app-sec:query-complexity,app-sec:time-complexity}, we have the following final result:

\begin{theorem}
    Let $M \in \CC^{2^n \times 2^n}$ satisfy the promise in~\cref{prob:main-problem} and fix $\delta \in (0,1)$ with failure budgets satisfying
    \begin{equation}
        \eta_{\mathcal{X}} + \defsum{x\in\mathcal{X}}{}{\eta_x} + \defsum{x\in\mathcal{X}_{\mathrm{deg}}}{}{\delta_x} \leq \delta.
    \end{equation}
    A convenient choice are those values given in~\cref{eqn:failure-probabilities}. Then,~\cref{alg:full-algo} determines the Pauli decomposition of $M$ with probability at least $1-\delta$. If we take the trivial bound $p^{\max}_x = k$, then the query complexity is
    \begin{equation}
        \bigOh{nk^2\log{\parens*{\frac{k}{\delta}}} + k^3\log^2{\parens*{\frac{k}{\delta}}}}
    \end{equation}
    and the runtime is
    \begin{equation}
        \bigOh{n^2k^2\log{\parens*{\frac{k}{\delta}}} + nk^3\log^2{\parens*{\frac{k}{\delta}}}}.
    \end{equation}
    If $k = \poly(n)$, then the total query complexity and the total runtime complexity remain polynomial in $n$, $k$, and $\log(1/\delta)$.
\end{theorem}

\begin{proof}
    By~\cref{lemma:x-bit-string-discovery}, the subroutine to find the $x$ bit strings succeeds with probability at least $1 - \eta_{\mathcal{X}}$. Conditioned on this process succeeding, we run~\cref{alg:unique-x-bit-strings} for every $x \in \mathcal{X}$, where by~\cref{thm:unique-algo-correctness}, each unique $x$ bit string is decoded correctly and each degenerate $x$ bit string is passed along to~\cref{alg:determine-degenerate-x-pauli-string}. This procedure correctly determines each degenerate $x$ bit string with probability at least $1-\delta_x$, as shown in~\cref{thm:fixed-x-recovery-probability}. Taking a union bound over all of these subroutines, the total failure probability is bounded by
    \begin{equation}
        \eta_{\mathcal{X}} + \defsum{x\in\mathcal{X}}{}{\eta_x} + \defsum{x\in\mathcal{X}_{\mathrm{deg}}}{}{\delta_x} \leq \delta.
    \end{equation}
    Thus, the whole algorithm finds the full Pauli decomposition with probability at least $1-\delta$. Finally, the query and runtime complexity bounds are proven according to~\cref{thm:query-complexity,thm:runtime-complexity}, respectively.
\end{proof}
\end{document}